\documentclass[letterpaper,11pt]{article}
\usepackage[utf8]{inputenc}

\usepackage{fancybox}

\usepackage{amsmath, amsthm, amssymb, thm-restate}
\usepackage{algorithmicx}
\usepackage[table,xcdraw]{xcolor}
\usepackage[vlined,linesnumbered]{algorithm2e}
\usepackage{setspace}
\usepackage{mathtools}
\usepackage{bbm}
\usepackage[numbers]{natbib}
\usepackage{comment} 
\usepackage[breakable]{tcolorbox}
\usepackage{xfrac}
\usepackage{multirow}
\usepackage{caption}
\usepackage{bm}
\usepackage{newfloat}
\usepackage{enumitem}
\usepackage{dblfloatfix} 
\usepackage{wrapfig}

\usepackage{fancyhdr}

\usepackage[margin=1in]{geometry}

\allowdisplaybreaks

\definecolor{mygreen}{RGB}{10,180,120}
\definecolor{myred}{RGB}{140,50,40}

\usepackage[linktocpage=true,
	pagebackref=true,colorlinks,
	linkcolor=myred,citecolor=mygreen,
	bookmarks,bookmarksopen,bookmarksnumbered]
	{hyperref}
\usepackage[noabbrev,nameinlink]{cleveref}

\usepackage{mdframed}

\newcommand{\rs}{Ruzsa-Szemer\'{e}di\xspace}

\newcommand{\Ot}{\ensuremath{\widetilde{O}}}
\newcommand{\eps}{\ensuremath{\varepsilon}}
\newcommand{\Paren}[1]{\Big(#1\Big)}

\newcommand{\paren}[1]{\ensuremath{\left(#1\right)}\xspace}
\newcommand{\card}[1]{\vert{#1}\vert}

\newcommand{\set}[1]{\ensuremath{\left\{ #1 \right\}}}
\newcommand{\poly}{\mbox{\rm poly}}
\newcommand{\polylog}{\mbox{\rm  polylog}}

\newcommand{\opt}{\textnormal{\ensuremath{\mbox{opt}}}\xspace}

\newcommand{\logstar}[0]{\ensuremath{\log^*}}

\renewcommand{\epsilon}{\varepsilon}

\DeclareMathOperator{\E}{\ensuremath{\normalfont \textbf{E}}}
\DeclareMathOperator{\supp}{\ensuremath{\normalfont supp}}

\newcommand{\hiddencomment}[1]{}

\newcommand{\RS}[2]{\ensuremath{\textnormal{\textsf{RS}}(#1,#2)}}
\newcommand{\MC}[2]{\ensuremath{\textnormal{\textsf{MC}}(#1,#2)}}

\newcommand{\topt}{\ensuremath{{\opt}}}

\newcommand{\DS}{\ensuremath{\mathbf{D}}}

\newcommand{\add}{\ensuremath{\textnormal{\texttt{insert}}}}
\newcommand{\member}{\ensuremath{\textnormal{\texttt{member}}}}

\renewcommand{\leq}{\leqslant}
\renewcommand{\geq}{\geqslant}

\renewcommand{\b}[1]{\ensuremath{\bm{\mathrm{#1}}}}

\usepackage[compact]{titlesec}
\usepackage[noabbrev,nameinlink]{cleveref}
\crefname{lemma}{Lemma}{Lemmas}
\crefname{theorem}{Theorem}{Theorems}
\crefname{property}{Property}{Properties}
\crefname{claim}{Claim}{Claims}
\crefname{result}{Result}{Results}
\crefname{definition}{Definition}{Definitions}
\crefname{observation}{Observation}{Observations}
\crefname{proposition}{Proposition}{Propositions}
\crefname{assumption}{Assumption}{Assumptions}
\crefname{line}{Line}{Lines}
\crefname{figure}{Figure}{Figures}
\creflabelformat{property}{(#1)#2#3}
\crefname{equation}{}{}
\crefname{section}{Section}{Sections}
\crefname{appendix}{Appendix}{Appendices}
\crefname{algCounter}{Algorithm}{Algorithms}
\Crefname{algCounter}{Algorithm}{Algorithms}

\newtheorem{theorem}{Theorem}
\newtheorem{lemma}{Lemma}[section]
\newtheorem{proposition}[lemma]{Proposition}
\newtheorem*{proposition*}{Proposition}
\newtheorem{corollary}[lemma]{Corollary}

\newtheorem{definition}[lemma]{Definition}
\newtheorem{claim}[lemma]{Claim}
\newtheorem*{claim*}{Claim}
\newtheorem{fact}[lemma]{Fact}

\newtheorem{remark}[lemma]{Remark}

\definecolor{mylightgray}{RGB}{230,230,230}

\algnewcommand{\IIf}[2]{\textbf{if} #1 \textbf{then} #2}
\algnewcommand{\EndIIf}{\unskip\ \algorithmicend\ \algorithmicif}

\newenvironment{whitetbox}{
\par\addvspace{0.1cm}
\begin{tcolorbox}[width=\textwidth,
                  boxsep=5pt,
                  left=1pt,
                  right=1pt,
                  top=2pt,
                  bottom=2pt,
                  boxrule=1pt,
                  arc=0pt,
                  colframe=black,
                  colback=white
                  ]
}{
\end{tcolorbox}
}

\newenvironment{myproof}{
\vspace{-0.5cm}
\begin{proof}
}{
\end{proof}
}

\newcounter{algCounter}

\def\eps{\epsilon}
\def\sizeof#1{\left|#1  \right|}

\newenvironment{fminipage}%
  {\begin{Sbox}\begin{minipage}}%
  {\end{minipage}\end{Sbox}\fbox{\TheSbox}}

\def\union{\cup}
\def\intersect{\cap}
\def\Union{\bigcup}

\newtheorem{mdresult}{Result}
\newenvironment{result}{\begin{mdframed}[backgroundcolor=lightgray!40,topline=false,rightline=false,leftline=false,bottomline=false,innertopmargin=2pt]\begin{mdresult}}{\end{mdresult}\end{mdframed}}

\theoremstyle{definition}

\tcbuselibrary{skins,breakable}
\tcbset{enhanced jigsaw}

\newtcolorbox{myalgbox}[1][]{
    enlarge top by=5pt,
    enlarge bottom by=5pt,
    breakable,
    boxsep=0pt,
    left=4pt,
    right=4pt,
    top=10pt,
    arc=0pt,
    boxrule=1pt,toprule=1pt,
    colback=white
}

\newtheorem{mdalg}{Algorithm}
\newenvironment{Algorithm}{
\begin{myalgbox}
\begin{mdalg}
}{
\end{mdalg}
\end{myalgbox}
}

\makeatletter
\renewcommand{\paragraph}{%
  \@startsection{paragraph}{4}%
  {\z@}{10pt}{-1em}%
  {\normalfont\normalsize\bfseries}%
}
\makeatother

\makeatletter
\patchcmd{\@algocf@start}
  {-1.5em}
  {0pt}
  {}{}
\makeatother

\title{On Regularity Lemma and Barriers in Streaming and Dynamic~Matching} 

 \author{
 Sepehr Assadi\footnote{(\texttt{sepehr@assadi.info}) Department of Computer Science, Rutgers University. Research supported in part by a NSF CAREER Grant CCF-2047061, a Google Research gift,
 and a Fulcrum award from Rutgers Research Council.} \and 
 Soheil Behnezhad\footnote{(\texttt{soheil.behnezhad@gmail.com}) Department of Computer Science, Stanford University.} \and
 Sanjeev Khanna\footnote{(\texttt{sanjeev@cis.upenn.edu}) Department of Computer and Information Science, University of Pennsylvania. Research supported in part by
NSF awards CCF-1763514, CCF-1934876, and CCF-2008305.}  \and
 Huan Li\footnote{(\texttt{huanli@cis.upenn.edu}). Department of Computer and Information Science, University of Pennsylvania.} 
 }

\date{}

\begin{document}

\maketitle

\pagenumbering{roman}

\begin{abstract}
\medskip

    We present a new approach for finding matchings in dense graphs by building on Szemer\'edi's celebrated Regularity Lemma. This allows us to obtain non-trivial albeit slight improvements over longstanding bounds for matchings in streaming and dynamic graphs. In particular, we establish the following results for $n$-vertex graphs: 
    
    \begin{itemize}[leftmargin=15pt]
        \item A deterministic \textbf{single-pass streaming} algorithm that finds 
        a $(1-o(1))$-approximate matching in $o(n^2)$ bits of space. This constitutes the first single-pass algorithm for this problem in sublinear space that improves over the $\sfrac{1}{2}$-approximation of the greedy algorithm. 
        
        \item A randomized \textbf{fully dynamic} algorithm that with high probability maintains a $(1-o(1))$-approximate matching in $o(n)$ worst-case update time per each edge insertion or deletion. The algorithm works even against an adaptive adversary. This is the first $o(n)$ update-time dynamic algorithm with approximation guarantee arbitrarily close to one. 
    \end{itemize}

     Given the use of regularity lemma, the improvement obtained by our algorithms over trivial bounds is only by some $(\logstar{n})^{\Theta(1)}$ factor. Nevertheless, in each case, they show that the ``right'' answer to the problem is not what is dictated by the previous bounds.  
     
     \medskip
     Finally, in the  streaming model, we also present a randomized $(1-o(1))$-approximation algorithm whose space can be {upper bounded} by the density of certain Ruzsa-Szemer\'edi (RS) graphs. While RS graphs by now have been used extensively to prove streaming lower bounds, ours is the first to use them as an upper bound tool for designing improved streaming algorithms.

\end{abstract}

\clearpage
\clearpage

\setcounter{tocdepth}{2}
\tableofcontents

\clearpage

\pagenumbering{arabic}
\setcounter{page}{1}

\section{Introduction}

Given a graph $G=(V,E)$, a matching $M$ in $G$ is any collection of edges that share no endpoints. Finding maximum matchings has been a cornerstone of algorithm design starting from the work of \citet*{Konig1916} over a century ago. Nevertheless, many fundamental questions regarding the complexity of this problem have remained unresolved, specifically in modern models of computations such as streaming or dynamic graphs. Indeed, in both mentioned models, despite significant attention, there has been {no} improvement in certain key cases over longstanding barriers that have remained in place since the introduction of the model itself. 

In this paper, we make an ever so slight improvement over these barriers, showing that the \emph{right} answer to the problem must be different than what is dictated by prior bounds. Our results combine tools from extremal combinatorics, primarily Szemer\'edi's Regularity Lemma~\cite{szemeredi1975regular} and its extensions, with multiple ideas (old and new) tailored to each model specifically. To put our results in more context, we start with the history of the problem in each model separately. 

\paragraph{Graph Streaming.} In this model, edges of an $n$-vertex graph appear one by one in a stream in an arbitrary order. The algorithm can read the edges in the arrival order while using a limited memory smaller than the input size, and output the solution at the end of the stream. The holy grail of algorithms here is a one that uses $O(n \cdot \polylog{(n)})$ memory and a single pass over the stream. The study of graph streaming algorithms were initiated by~\citet*{FeigenbaumKMSZ05} who already observed that there is a straightforward $1/2$-approximation algorithm for matching in $O(n\log{n})$ space\footnote{Throughout, as is standard, we always measure the space of streaming algorithms in \emph{bits}.}: greedily maintain a maximal matching in the stream. They further proved that finding an exact maximum matching requires $\Omega(n^2)$ space, which matches the trivial algorithm that stores the entire input via its adjacency matrix.  

Almost two decades since~\cite{FeigenbaumKMSZ05}, there are still no better algorithms for matchings than these two straightforward solutions. On the lower bound front, a series of work by~\citet*{GoelKK12} and~\citet*{Kapralov13} culminated in a recent work of~\citet*{Kapralov21} that rules
out better than $1/(1+\ln{2}) \approx 0.59$ approximation in  $n^{1+o(1/\log\log{n})}$ space. This lack of progress has led researchers to consider various relaxations of the problem, in particular by allowing a few more passes over the
input (e.g., in~\cite{KonradMM12,KaleT17,Konrad18,KonradN21,Assadi22,FeldmanS22}) or assuming random arrival of edges in the stream (e.g., in~\cite{KonradMM12,AssadiBBMS19,FarhadiHMRR20,Bernstein20,AssadiB21})\footnote{See the papers of~\citet*{FeldmanS22} and~\citet*{AssadiB21}, respectively, for the state of the art in each case, and more details on previous work on each relaxation.}.  
At this point, beating $1/2$-approximation factor of the greedy algorithm in $O(n \cdot \polylog{(n)})$ space, or even much larger than that, has become one of the most central open questions of the graph streaming literature; see, e.g.,~\cite{KonradMM12,McGregor14,Kapralov21,wajc2020matching,FeldmanS22} for various references to this question. 

\paragraph{(Fully) Dynamic Graphs.} In this model, we have an $n$-vertex graph that undergoes an arbitrary sequence of edge insertions and deletions. The goal is to maintain the solution to the problem, say an approximate maximum matching of the graph, with a quick update time per each insertion or deletion. Dynamic algorithms for matchings were studied first in this model by~\citet*{IvkovicL93} in 1993 and continue to be a highly active area of research (see, e.g.~\cite{OnakR10,BaswanaGS11,BaswanaGS18,NeimanS-STOC13,GuptaP13,BhattacharyaHI-SiamJC18,BernsteinS16,BhattacharyaHN-SODA17,BhattacharyaHN-STOC16,Solomon-FOCS16,CharikarS18-ICALP,ArarCCSW-ICALP18,BernsteinFH-SODA19,BehnezhadDHSS-FOCS19,BehnezhadLM-SODA20,Wajc-STOC20,BernsteinDL-STOC21,BhattacharyaK21-ICALP21,RoghaniSW-ITCS22,Kiss-ITCS22,GrandoniSSU-SOSA22,BehnezhadK22} and references therein).

There is a folklore algorithm that for any $\eps > 0$, maintains a $(1-\eps)$-approximate matching in $O(n/\eps^2)$ (amortized) update time: Assume inductively that we have a $(1-\eps/2)$-approximate matching $M$ of the current graph; $(i)$ for the next $(\eps/2) \cdot \card{M}$ updates do nothing and return $M$ still as the answer; after that, $(ii)$ compute a $(1-\eps/2)$-approximate matching of the current graph in $O(m/\eps)$ time using the Hopcroft-Karp algorithm~\cite{HopcroftK73} where $m$ is the number of edges in the graph and repeat from step $(i)$. Since $m = O(n \cdot \card{M})$ in any graph with maximum matching size bounded by $O(\card{M})$, 
the amortized update time will be $O(n/\eps^2)$, and the correctness can be easily verified. This algorithm can also be deamortized using standard batching ideas. 

For sparser graphs, this folklore algorithm was improved by~\citet*{GuptaP13} to achieve an $O(\sqrt{m}/\eps^2)$ update time where $m$ denotes the (dynamic) number of edges. Faster algorithms are only known for smaller approximations between 1/2 and 2/3 which can respectively be maintained in $O(1)$ \cite{Solomon16} (see also~\cite{BaswanaGS11}) and $O(\sqrt{n})$ update-times \cite{BernsteinS16}.
See also a recent result of \cite{BehnezhadK22} for update-time/approximation trade-offs between 1/2 and 2/3. Yet, for the original $(1-\eps)$-approximation question, raised e.g. in~\cite{GuptaP13}, an $O(n)$ update time still remains a barrier in general.

\subsection{Our Contributions} 

We present the first algorithms that beat the aforementioned barriers for finding matchings in streaming and dynamic graphs with non-trivial albeit quite small factors:

\begin{result}[Formalized in~\Cref{thm:stream1}] \label{res:stream1}
There is a randomized $(1-o(1))$-approximate matching algorithm in single-pass streams with adversarial order of edge arrivals in $n^2/(\logstar{n})^{\Omega(1)}$ space and polynomial time.   
\end{result}
 
    This is the first $o(n^2)$-space algorithm for matchings in adversarial-order streams with better than $1/2$-approximation guarantee. In fact, it was not known previously how to achieve a $(1-o(1))$-approximation 
    in $o(n^2)$ space even on random-arrival streams and 
    even if we allow any constant number passes over the input (see~\cite{McGregor05,AhnG11,AhnG18,AssadiLT21,AssadiJJST22,FischerMU22} for representative examples of multi-pass streaming matching algorithms\footnote{The state-of-the-art is the $O_{\eps}(n^{1+1/p})$-space $O(p/\eps)$-pass 
    algorithm by~\citet*{AhnG18} and $O_{\eps}(n \cdot \polylog{(n)})$-space $\poly{(1/\eps)}$-pass by~\citet*{FischerMU22} (see also algorithms by~\citet*{AssadiJJST22} and~\citet*{AssadiLT21} with improved bounds for bipartite graphs).}).
    Moreover, combined with the lower bound of $\Omega(n^2)$ space by~\cite{FeigenbaumKMSZ05} for computing exact matchings, \Cref{res:stream1} shows the first provable separation between the space complexity of computing nearly-optimal versus exact-optimal matchings in single-pass streams.
    
\begin{result}[Formalized in~\Cref{thm:dynamicalgo}]\label{res:dynamic}
There is a randomized $(1-o(1))$-approximate matching algorithm in fully dynamic graphs against an adaptive adversary with $n/(\logstar{n})^{\Omega(1)}$ worst case update time. 
\end{result}
This is the first algorithm for matchings in fully dynamic graphs that achieves $o(n)$ update time for all densities with close to one approximation guarantee (this was not known before even for oblivious adversaries). 

The key idea behind both these results is to maintain a \emph{matching cover}---introduced by~\citet*{GoelKK12} in spirit of cut/spectral sparsifiers---that is a ``sparse'' subgraph which approximately preserve matchings in each induced subgraph of the input graph. 
We present a polynomial time algorithm for constructing $o(n^2)$-size matching covers using Szemer\'edi's Regularity Lemma~\cite{szemeredi1975regular} 
and along the way extend them to general graphs (\cite{GoelKK12} only proves their existence and for bipartite graphs). 
We then show this new construction can be maintained in streaming and dynamic graphs using several new ideas combined  with standard tools from prior work specific to each model. We elaborate more on our techniques in~\Cref{sec:overview}. 

We also present a third result specific to the graph streaming model. \emph{All} previous lower bounds for approximating matchings in graph streams in a single pass~\cite{GoelKK12,Kapralov13,AssadiKL17,Kapralov21}, multi-pass~\cite{AssadiR20,ChenKPSSY21,Assadi22}, or random-order streams~\cite{AssadiB21} rely on constructions based on \rs (RS) graphs \cite{RuszaS78}. These are graphs whose edges can be partitioned into ``large'' \emph{induced} matchings (see~\Cref{sec:rs} for details). We present a converse approach  by developing a streaming algorithm for matchings whose space can be \emph{upper bounded} by the density of (certain) RS graphs. In particular, 

\begin{result} \label{res:stream2}\!\!\textnormal{(Formalized in~\Cref{thm:stream2})}\!
For any $k \geq 1$, there is a randomized {\small $(1-o(1))$}-approximate matching algorithm in single-pass streams with adversarial order of edge arrivals in
\[
\paren{n^2/k + RS(n,o(n/k))} \cdot \poly\!\log{(k)}
\]
space and exponential time; here, $RS(n,r)$ denotes the largest number of edges in any $n$-vertex graph whose edges can be partitioned
into induced matchings of size $r$. The algorithm can be made deterministic if the goal is an additive $o(n)$ approximation instead. 
\end{result}

\Cref{res:stream2} builds on and  generalize the RS graph based communication protocol of~\citet*{GoelKK12} to the streaming model (and from bipartite to general graphs). 

To put this result in more context, notice that RS graphs are naturally becoming denser and denser by reducing the size of their induced matchings\footnote{Any (simple) graph can be seen as an RS graph with induced matchings of size one.}, leading to a tradeoff between the two terms in the space guarantee 
of~\Cref{res:stream2}. Unfortunately, proving tight bounds on the density of RS graphs is a notoriously difficult problem in combinatorics (see, e.g.,~\cite{gowers2001some,conlon2013graph,fox2017graphs}). As such,
the space complexity of the algorithm in~\Cref{res:stream2} as purely a function of $n$ is not clear at this point. However, using~\Cref{res:stream2} 
combined with Fox's triangle-removal lemma~\cite{Fox11} that, to our knowledge, provides the best approach currently for bounding density of RS graphs with $o(n)$-size induced matchings, we can obtain the following result: 
\vspace{-3pt}
\begin{itemize}[leftmargin=10pt]
    \item \textbf{A corollary of~\Cref{res:stream2}} (Formalized in~\Cref{cor:stream2})\textbf{.}
    There is a deterministic $(1-o(1))$-approximate matching algorithm in single-pass streams with adversarial order of edge arrivals using $n^2/2^{\Theta(\logstar{n})}$ space and exponential time.  
\end{itemize}
\vspace{-3pt}

This corollary improves upon our algorithm in~\Cref{res:stream1} based on the regularity lemma in terms of approximation ratio and being deterministic at the cost of taking exponential time.
Moreover, by a result of~\citet*{GoelKK12} on lower bounds for streaming matching via RS graphs, obtaining streaming algorithms with better space complexity than this corollary, namely, beating $n^2$ by more than a $2^{\Theta(\logstar{n})}$ factor, immediately implies improved RS graph upper bounds; in other words, improving upon our algorithm at the very least requires proving better RS graph upper bounds than currently known bounds (see~\cite{FoxHS15} for why this is a challenging task).  

Finally, given the current state of knowledge about RS graphs (see~\cite{AlonMS12,fox2017graphs}), it is possible that the space of the algorithm in~\Cref{res:stream2} can be improved to $n^{2}/2^{\Theta(\sqrt{\log{n}})}$---thus more than any $\polylog{(n)}$ factor shaving in the space over $n^2$---assuming that the currently best construction of dense RS graphs in~\cite{RuszaS78} (see also~\cite{AlonMS12}) with induced matchings of size $n/2^{\Theta(\sqrt{\log{n}})}$ cannot be improved substantially to larger induced matching sizes. 

\medskip

In conclusion, our paper shows that these longstanding barriers in computing large matchings in streaming and dynamic graphs can at least be broken by some non-trivial albeit quite small factors. Moreover, these algorithms rely on techniques and ideas that are vastly different from prior approaches used in these two models. We hope our work paves the path toward further progress on these longstanding open questions. 

\section{Technical Overview}\label{sec:overview} 

{\em Matching sparsifiers}, which loosely speaking, are sparse subgraphs that approximately preserve the maximum matching have long been known to be an important tool for fully dynamic and (variants of) streaming algorithms. Some prominent examples include edge-degree constrained subgraphs (EDCS) \cite{BernsteinS15,BernsteinS16} and its generalizations \cite{AssadiB19,BehnezhadK22}, kernels \cite{BhattacharyaHN-STOC16,BhattacharyaHI-SiamJC18,ArarCCSW-ICALP18,BernsteinDL-STOC21}, and matching skeletons \cite{GoelKK12}. One of our main contributions, and the key to both  \Cref{res:stream1} and \Cref{res:dynamic}, is a new matching sparsifier based on Szemer\'edi's Regularity Lemma. 

Our matching sparsifier, more strongly, is a {\em matching cover}---\`a la \citet*{GoelKK12}---which not only preserves an approximate {\em maximum} matching of the graph, but rather ``covers'' smaller matchings of it as well. Let us formalize this. For a given graph $G$,
we write $\mu(G)$ to denote the maximum matching size
of $G$, and write $G[A,B]$ to denote the bipartite subgraph
of $G$ between some disjoint vertex subsets $A,B$.
We say a subgraph $H$ of $G$ is an \textbf{$\alpha$-matching cover}
for  $\alpha \in (0,1)$
if for any disjoint subsets of vertices $(A,B)$ in $G$,
$\mu(H[A,B]) \geq \mu(G[A,B]) - \alpha \cdot n$. That is, $H$ preserves the  largest matching in the induced bipartite subgraph $G[A, B]$
to within an additive $\alpha \cdot n$ factor. While from an information theoretic perspective, {\em existence} of an $o(n^2)$-edge $o(1)$-matching cover for bipartite graphs was proved in the original paper of \citet*{GoelKK12}, it was not known up until now whether one can find such matching covers efficiently, say in polynomial time. Note that this is specially important, for instance, for applications in dynamic algorithms where the goal is to optimize the update time.

In this paper, we prove that there is an $\tilde{O}(n^{\omega})$-time\footnote{Here and throughout,
$\omega\approx 2.37286$ is the matrix multiplication exponent with current best bounds achieved by~\citet*{AlmanW21}. } offline
algorithm that computes an $o(n^2)$-edge $o(1)$-matching cover of any $n$-vertex graph (not necessarily bipartite). Our algorithm builds on Szemer\'edi's Regularity Lemma (and its algorithmic version due to \citet*{AlonDLRY92}). We first explain how our offline algorithm for obtaining a matching cover works, and then outline its use in obtaining improved dynamic matching and streaming algorithms.

\subsection{Matching Covers via Regularity Lemma}
  Roughly speaking, Szemer\'edi's regularity lemma~\cite{szemeredi1975regular} says that the vertices of any graph can be partitioned
  into a small {\em irregular} part $C_0$ with $|C_0| = o(n)$, plus $k$ other equal-size parts
  $C_1,\ldots,C_k$ for some $k \in [\omega(1),\log n]$.
  The latter $k$ parts have the property that all but $o(1)$-fraction of the
  $C_i,C_j$ pairs are {\em regular}: for any pair of subsets
  $X\subseteq C_i, Y\subseteq C_j$ with large enough size, 
  the edge density between $X,Y$ is similar to that of $C_i,C_j$.
  Therefore, if the edges between $C_i,C_j$ are dense to start with, the density will also be high between every large enough $X\subseteq C_i, Y\subseteq C_j$ pair.

 It is not difficult to see that by regularity, any large matching between a dense regular pair $C_i,C_j$ can be mostly preserved if we subsample  edges between them at a sufficiently high rate $p = o(1)$. In particular, the subsampled graph will be a matching cover of the graph induced by edges between the regular pair $C_i,C_j$. 
 This suggests a natural strategy for building an $o(1)$-matching cover with $o(n^2)$ edges: {\em subsample} the edges between dense regular $C_i,C_j$ pairs at rate $p = o(1)$ and take {\em all} other edges. We would like to show that this is an $\alpha$-matching cover for some $\alpha=o(1)$.

 This idea runs into the following problem. Suppose we have an 
  $(\alpha n)$-size matching $M$ whose edges are evenly distributed across
  all $\binom{k}{2}$ pairs of $C_i,C_j$,
  then the number of edges of $M$ between each $C_i,C_j$ pair is only $O(\alpha \cdot n / k^2)$.
  This means that only an $O(\alpha/k)$ fraction of vertices in $C_i,C_j$ are matched to each other -- this is unfortunately way too small to invoke the regularity property. 

We get around this issue by first focusing on solving an \emph{$\alpha$-hitting set} problem: find one edge between endpoints of any $(\alpha n)$-size matching -- we will show later on using a similar argument as in~\cite{GoelKK12} that this is sufficient for obtaining an $\alpha$-matching cover. Now to fix our problem about an $(\alpha n)$-size matching whose edges are distributed across many pairs, we 
present a strategy for {\em consolidating the support} of a matching over different pairs. This consolidation argument shows that whenever there is a large matching $M$ between dense regular $C_i,C_j$ pairs,
  there must also exist another (almost as) large matching $M'$ that is supported on the same set of vertices $V(M)$ but
  only uses edges between a small number of such $C_i,C_j$ pairs.
  As a result, there must exist one pair of $C_i,C_j$ where a substantial fraction
  of vertices are matched to each other, to which we are now able to apply regularity to prove the existence
  of an edge between them in the subsampled graph (which solves our $\alpha$-hitting set problem).
  At a high level, our argument for consolidating the support of the matching is proved by $(i)$ viewing the matching $M$ as a {\em fractional} matching
  in a meta graph obtained by contracting each $C_i$ into a supernode; and
  $(ii)$ rounding the fractional matching by an edge sampling process. 

All in all, using the algorithm of~\cite{AlonDLRY92} for finding the regularity lemma partition in $\Ot(n^{\omega})$ time, and a direct sampling algorithm between dense regular pairs, this step gives us an $\Ot(n^{\omega})$ time and $\Ot(n)$ space algorithm for finding an $\alpha$-matching cover of size $n^2/(\logstar{n})^{\Omega(1)}$ for some $\alpha=1/(\logstar{n})^{\Omega(1)}$.

\subsection{Applications of Matching Cover}

\paragraph{A fully dynamic matching algorithm.}  
The matching cover algorithm above is offline. But observe that since the algorithm takes $\tilde{O}(n^\omega) = n^{3-\Omega(1)}$ time, the time spent {\em per edge} in a dense instance is sublinear in $n$. This gives hope that perhaps such a matching cover can be maintained in $o(n)$ time, and indeed we show this to be the case.

Our algorithm roughly proceeds by re-computing
an $o(1)$-matching cover every $\tilde{\Theta}(n^{\omega-1})$ updates,
and then using the $O(\sqrt{m})$-update time data structure
by~\citet*{GuptaP13} to maintain a nearly optimal
matching in the matching cover
through the subsequent $\tilde{\Theta}(n^{\omega-1})$ updates.
Since the matching cover only has $o(n^2)$ edges,
we immediately get an update time of $o(n)$ for the Gupta-Peng algorithm.
To argue the correctness, we show that the matching cover found
by our offline algorithm has the additional
feature that it is {\em robust to edge updates}:
not only is it an $o(1)$-matching cover of the graph
at the time we compute it, but it remains
an $o(1)$-matching cover throughout any arbitrary sequence of $n^{2-o(1)}$ updates.
This suffices to show that our algorithm can dynamically maintain
an approximate matching with an additive error $o(n)$.

When the number of edges is close to $n^2$,
this additive approximation becomes a $(1-o(1))$-multiplicative approximation,
since the maximum matching size is itself $\Omega(n)$.
On the other hand, when the number of edges is $o(n^2)$, directly applying
the Gupta-Peng data structure gives us a nearly-optimal matching in $o(n)$ update time.
Our final algorithm then balances the dense and the sparse regimes together
to maintain a $(1-o(1))$-approximate matching in $o(n)$ update time.

\paragraph{Streaming algorithms.} 
Our streaming algorithm in~\Cref{res:stream1} is also based on using matching covers as a natural sparsifier for matchings. The algorithm works 
through a series of buffers of edges $B_1,B_2,\ldots,$. The first buffer $B_1$ reads the edges from the input until it gets ``full'', i.e., receives some $o(n^2)$ edges (which is some constant factor larger than the size of our matching cover). At that point we compute a matching cover of the edges in the buffer using an offline/non-streaming 
algorithm and send its edges to the buffer $B_2$; then, we ``restart'' $B_1$ by emptying all its current edges and letting it collect more edges from the stream. The same approach is repeated across all other buffers as well. The number of these buffers can be bounded as only a constant fraction of edges
in one buffer can make their way to the next one, eventually reaching a buffer that never gets full. This also implies
that fewer edges will be be further ``sparsified'' in each matching cover, thus the error occurred due to the approximation guarantee of the matching cover does not get amplified ``too much''. Thus, using this algorithm along with our matching cover algorithm for regularity lemma, leads to an $o(n^2)$-space $(1-o(1))$-approximation algorithm for 
single-pass streaming matchings. 

The strategy we outlined above works for \emph{any} choice of matching cover (as long as we can compute it in a small space offline). Thus, 
we can alternatively implement the matching cover subroutine by simply enumerating all subgraphs of the input (in exponential time) and the \emph{optimal}
one. 
An argument due to~\citet*{GoelKK12}---extended in our paper to general graphs---shows that density of optimal matching covers can be bounded by the density of certain RS graphs. To obtain~\Cref{res:stream2}, we also need to turn the additive approximation guarantee of the matching cover into a multiplicative bound. This is done using vertex-sparsification methods of~\citet*{AssadiKLY16} and~\citet*{ChitnisCEHMMV16} (as specified in~\cite{AssadiKL16ec}) that reduce the number of vertices in the graph down to its maximum matching size without reducing the matching size by much. This turns the additive guarantee of the matching
cover into a multiplicative one, giving us~\Cref{res:stream2} as well. 

Finally, one key step in making the above algorithms work is to store the $o(n^2)$ edges they have in the buffers more efficiently than spending $\Theta(\log{n})$ bits per each (which is prohibitive for us given the extremely small improvement in the space the algorithms get over the trivial $O(n^2)$ bound). This is done by storing the edges via the succinct dynamic dictionary of~\citet*{RamanR03} (see~\Cref{sec:succinct}) and then performing all the computation in this compressed space instead.

\section{Preliminaries}\label{sec:prelim} 

\paragraph{Notation.} For any integer $t\geq s \geq 1$, we let $[t]:=\set{1,\ldots,t}$ and let $[s,t] = \{s, \ldots, t\}$. 
We use the term with high probability, abbreviated w.h.p., to imply probability at least $1-1/n^c$ for any desirably large constant $c \geq 1$ (that might affect the hidden constants in our statements).

For a graph $G=(V,E)$, we use $V(G) = V$ to denote the set of vertices and $E(G) = E$ to denote the edges. For any subsets of edges $F \subseteq E$ and disjoint subsets of vertices $X,Y \subseteq V$, we use $X(F)$ and $Y(F)$ to
denote the edges of $F$ incident on $X$ and $Y$, respectively, and $F(X,Y)$ to denote the edges of $F$ going between $X$ and $Y$.  Similarly, we use $G[X]$ and $G[X,Y]$ to respectively denote the subgraph of $G$ induced on vertices $X$, and the bipartite subgraph of $G$ between vertices $X$ and $Y$. For any $p \in [0, 1]$, we use $G[p]$ to denote a random subgraph of $G$ that includes each edge of $G$ independently with probability $p$.

For any graph $G$, $\mu(G)$ denotes the size of the maximum matching in $G$. We have, 

\begin{fact}\label{fact:m-mu(G)}
    Any graph $G$ has at most $2n \cdot \mu(G)$ edges. 
\end{fact}

The proof of~\Cref{fact:m-mu(G)} is simply based on picking an arbitrary edge of the graph and adding to a matching, removing at most $2n$ edges incident on this edge, and repeating until the graph is empty.   

We will also need the following version of Hall's theorem.

\begin{proposition}[Extended Hall’s marriage theorem; cf.~\cite{Hall87}]\label{prop:hall}
Let $G=(L,R,E)$ be any bipartite graph with $|L|=|R|=n$.
Then $\max(|A|-|N_G(A)|) = n - \mu(G)$, where
$A$ ranges over all subsets of $L$ and $R$, and
$N_G(A)$ denotes the neighbors of $A$ in $G$.
\end{proposition}

\subsection{Szemer\'edi's Regularity Lemma}\label{sec:regularity-lemma}

Szemer\'edi's Regularity Lemma~\cite{szemeredi1975regular} is a powerful tool in extremal combinatorics. Loosely speaking, it says that every dense graph can be well-approximated by a ``small'' collection  of random-like subgraphs. To formally state the lemma, we need a few definitions.

Let $G= (V,E)$ be any given graph,
and $A,B\subseteq V$ be any disjoint vertex subsets.
We write $e(A,B)$ to denote the number of edges between $A,B$.
If $A,B\neq \emptyset$, we define the \textbf{density} of edges between $A$ and $B$ by:
\begin{align*}
    d(A,B) := \frac{e(A,B)}{|A||B|}.
\end{align*}
For a parameter $\gamma \in (0,1)$,
we say $(A,B)$ is \textbf{$\gamma$-regular} if for every $X\subseteq A$ and
$Y\subseteq B$ satisfying $|X|\geq \gamma \cdot |A|$ and $|Y|\geq \gamma \cdot |B|$,
we have
 $   \sizeof{d(A,B) - d(X,Y)} < \gamma.$

Let $C_0,C_1,\ldots,C_k$ be a partition of the vertex set $V$.
We say this partition is {\bf equitable} if the classes
$C_1,\ldots,C_k$ all have the same size.
We will call $C_0$ the {\bf exceptional class}.
We say this partition is $\gamma$-regular if
both of the following statements are true:
\begin{enumerate}
    \item It is equitable and $|C_0|\leq \gamma n$.
    \item All but at most $\gamma \binom{k}{2}$ of the pairs
    $C_i,C_j$ for $1\leq i < j \leq k$ are $\gamma$-regular.
\end{enumerate}

Instead of the original formulation of Szemer\'edi's  Regularity Lemma in~\cite{szemeredi1975regular}, we state an algorithmic version of it due to \citet*{AlonDLRY92}.

\begin{proposition}[\cite{AlonDLRY92}]
\label{thm:regularityalgo}
There exists a function
$Q:\mathbb{R}^{+}\times \mathbb{R}^{+}\to \mathbb{R}^{+}$
satisfying $\log^* Q(x,y) \leq \poly(x,y)$ for all $x,y$,
such that,
given any $n$-vertex graph $G=(V,E)$
and $\gamma \in (0,1), t\geq 1$,
one can find in $n^\omega \cdot Q(t,1/\gamma)$ time
a $\gamma$-regular partition of $V$
into $k+1$ classes such that $t\leq k \leq Q(t,1/\gamma)$.
\end{proposition}

The algorithm in~\Cref{thm:regularityalgo} can also be implemented in a space-efficient manner (which is needed for our streaming algorithms). See~\Cref{app:prop-regularity-space-efficient} for a proof sketch. 

\begin{proposition}\label{prop:regularity-space-efficient}
  Given query access to the adjacency matrix,
  the algorithm in~\Cref{thm:regularityalgo} can be implemented
  in $O(n \cdot Q(t,1/\gamma)\log n)$ space and $\poly(n, Q(t,1/\gamma))$ time.
\end{proposition}

\subsection{Fox's Triangle Removal Lemma}\label{sec:triangle-removal}

Similar to the Regularity Lemma, the Triangle Removal Lemma is another highly useful tool in extremal combinatorics. While original proofs of this lemma were based on the regularity lemma,~\citet*{Fox11} presented a proof that bypasses regularity lemma and thus obtains stronger bounds. 
We will use this result also in one of our streaming algorithms. 

\begin{proposition}[\cite{Fox11}]\label{prop:triangle-removal}
    There exists an absolute constant $b > 1$ such that the following is true. For any $\gamma \in (0,1)$ let $\delta := \delta(\gamma)$ be inverse of the tower of twos of height $b \cdot \log{(1/\gamma)}$, i.e., $\delta^{-1} = 2 \upuparrows b \cdot \log{(1/\gamma)}$.  Then, any $n$-vertex graph with at most $\delta \cdot n^3$ triangles can be made triangle-free by removing at most $\gamma \cdot n^2$ edges. 
\end{proposition}

\subsection{\rs Graphs}\label{sec:rs}

A matching $M$ in a graph $G$ is called an \textbf{induced matching} iff the subgraph of $G$ induced on vertices of $M$ only contains the edges of $M$ itself; in other words, 
there are no other edges between the vertices of this matching. 

\begin{definition}\label{def:rs}
For integers $r,t \geq 1$, a graph $G=(V,E)$ is called an \emph{$(r,t)$-\rs graph} (RS graph for short) iff its edge-set $E$ can be partitioned into $t$ \underline{induced} matchings $M_1,\ldots, M_t$, each of size $r$. For any integer $n \geq 1$ and parameter $\beta \in (0,1/2)$, we use $\RS{n}{\beta}$ to denote the maximum number of edges in any $n$-vertex RS graph with induced matchings of size $\beta \cdot n$. 
\end{definition}

RS graphs have been extensively studied as they arise naturally in property testing, PCP constructions, additive combinatorics, streaming algorithms, graph sparsification, etc. (see, e.g.,~\cite{BirkLM93,HastadW03,FischerLNRRS02,Alon02,TaoV06,AlonS06,AlonMS12,GoelKK12,FoxHS15,AssadiB19,KapralovKTY21}).  In particular, a line of work 
initiated by~\citet*{GoelKK12} have used different constructions of these graphs to prove communication complexity and streaming lower bounds for graph streaming algorithms~\cite{GoelKK12,Kapralov13,Konrad15,AssadiKLY16,AssadiKL17,CormodeDK19,AssadiR20,Kapralov21,AssadiB21,ChenKPSSY21}. In this work however, we shall use them as an upper bound tool. The only other upper bound application of these graphs in a similar context that we are aware of is the communication protocols of~\cite{GoelKK12}: they show that to obtain a one-way communication protocol for $\eps \cdot n$-additive approximation of matchings, roughly $O(\RS{n}{\eps})$ communication is sufficient and also necessary. 

We establish a simple property of the $\RS{n}{\beta}$ function in~\Cref{def:rs} that relates density of different RS graphs with similar parameters (see~\Cref{app:prop-rs-move} for the proof). 

\begin{claim}\label{clm:rs-move}
    For any integer $n \geq 1$ and real $0 < \beta < 1$,  $\RS{2n}{3\beta} \leq O(1) \cdot \RS{n}{\beta}$. 
\end{claim}

\subsection{Succinct Dynamic Dictionaries}\label{sec:succinct} 

We need to use succinct dynamic dictionaries from prior work in~\cite{BrodnikM99,Pagh01,RamanR03}. For concreteness, we use the construction of~\cite{RamanR03} although the other ones work as fine also
for us. 
\begin{proposition}[c.f.~\cite{RamanR03}]\label{prop:succinct-dict}
	There exists a dynamic data structure $\DS$ for maintaining a subset $S$ of size at most $s$ from a universe $U$ of size $u$  that supports the following operations: 
	\begin{itemize}
		\item $\DS.\,\add(a)$: Inserts an element $a \in U$ to the set $S$;
		\item $\DS.\,\member(a)$: Returns whether the given element $a \in U$ belongs to $U$ or not; 
	\end{itemize}
	The data structure requires $(1+o(1)) \cdot \log{{u}\choose{s}}$ bits of space to store $S$ and answers each query in $O(1)$ amortized expected time or $O(s)$ worst-case deterministic time.
\end{proposition}


\section{A Matching Cover via Regularity Lemma}\label{sec:matcover}

In this section, we give a polynomial time algorithm for constructing an {\em matching cover} of size $o(n^2)$. We use the algorithm of this section both in the streaming model and the dynamic model.

Let us start by formally defining matching covers.

\begin{definition}[\cite{GoelKK12}]
\label{def:matching-cover}
    A subgraph $H$ of an $n$-vertex graph $G$ is an $\alpha$-matching cover of $G$ if for any disjoint subsets of vertices $(A,B)$ in $G$, we have $\mu(H[A,B]) \geq \mu(G[A,B]) - \alpha \cdot n$. 
\end{definition}

The following theorem is our main result of this section. 

\begin{theorem}\label{thm:matching-cover-polytime}
Given any $n$-vertex graph $G$, for some
$\alpha = (\log^*n)^{-\Omega(1)}$, there is an $O(n^{\omega}\log n)$ time algorithm, formalized below as \Cref{alg:matching-cover}, for finding an $\alpha$-matching cover $H$ of $G$ with at most $n^2/(\logstar n)^{\Omega(1)}$ edges.
\end{theorem}

Even though existence of $o(n^2)$ size $o(1)$-matching covers for bipartite graphs was already proved by \citet*{GoelKK12}, it was not known whether it is possible to find one in polynomial time (nor whether they also exist for general, not necessarily bipartite, graphs).

\subsection{First Step: A Hitting Set Argument}

In this section, we give an algorithm for finding an {\em $\alpha$-hitting set}, defined below. We later show in \cref{sec:hitting-to-matching-cover} that this can be turned into a matching cover.

\begin{definition}\label{def:hitting-set}
    A subgraph $H$ of an $n$-vertex graph $G$ is an $\alpha$-hitting set of $G$ if for any disjoint subsets of vertices $(A, B)$ in $G$ satisfying $|A| = |B| = \alpha n$ and $\mu(G[A, B]) = \alpha n$, there is at least one edge between $A$ and $B$ in $H$.
\end{definition}

The following lemma is our main guarantee of this section.

\begin{lemma}\label{lem:hittingset}
Given any $n$-vertex graph $G$, for some $\alpha = (\log^*n)^{-\Omega(1)}$, there is an $O(n^{\omega}\log n)$ time algorithm, formalized below as \Cref{alg:matching-cover}, for finding an $\alpha$-hitting set $H$ of $G$ with at most $n^2/(\logstar n)^{\Omega(1)}$ edges.
\end{lemma}

\Cref{thm:regularityalgo} below formalizes our algorithm for \cref{lem:hittingset}.

\begin{Algorithm}\label{alg:matching-cover}
  \begin{center}
		The construction of the matching cover for \cref{thm:matching-cover-polytime}.
	\end{center}
  
  \vspace{0.3cm}
  
  Let $t \gets (\log^*n)^{\delta}$, $\gamma \gets (\log^*n)^{-\delta}$
  for some constant $\delta\in(0,1)$
  such that $Q(t,1/\gamma) \leq \log n$.

    \begin{enumerate}[leftmargin=25pt,label=(\roman*)]
    \item Run the algorithm in \cref{thm:regularityalgo}
    to obtain a $\gamma$-regular partition
    $C_0,\ldots, C_{k}$ with
    $t = (\log^* n)^{\delta} \leq k\leq Q(t,1/\gamma) \leq \log n$.
    \label{line:2a}
    \item Let $(C_i, C_j)$ for $i \not= j$. We say $(C_i, C_j)$ is {\em good} if (i) $i, j \not= 0$,  (ii) it is $\gamma$-regular, and (iii) $d(C_i,C_j) \geq 8\gamma$. Otherwise, we say $(C_i, C_j)$ is {\em bad}.
    \item Let $F\subseteq E$ be
    a subset of edges
    obtained by $F := F_1 \union F_2 \union F_3$ where
    \begin{enumerate}
        \item $F_1$ contains
        the edges between the bad pairs;
        i.e. $F_1:= \Union_{\text{bad}\ C_i,C_j} E\intersect (C_i\times C_j)$.
        \item $F_2$ contains
        the edges within each class;
        i.e. $F_2:= \Union_{0\leq i \leq k} E\intersect (C_i\times C_i)$.
        \item $F_3$ is obtained
        by sampling the edges
        between good pairs with probability
        $p := \frac{10}{\log n}$; i.e.
        $F_3 = \Union_{\text{good}\ C_i,C_j}
        (E\intersect (C_i\times C_j))[p]$.
    \end{enumerate}
    \item Return $F$.
\end{enumerate}
\end{Algorithm}

It is not hard to see that \Cref{alg:matching-cover} outputs a subgraph with $O(\gamma n^2) = o(n^2)$ edges, since essentially the dense parts of the decomposition are subsampled and there are `few' other edges in the graph. The following claim formalizes this.

\begin{claim}\label{cl:fewedgespicked}
    The output $F$ of \Cref{alg:matching-cover}, w.h.p., has at most $O(\gamma n^2)$ edges.
\end{claim}
\begin{myproof}
    Let us first count the number of edges in $F_1$ and $F_2$. We do so by counting the number of different types of bad edges separately.
  \begin{enumerate}
    \item \textbf{Edges within each class:}
      The total number is bounded by
      \begin{align*}
        \binom{|C_0|}{2} +
        \sum_{i=1}^{k} \binom{|C_i|}{2}
        \leq \binom{\gamma n}{2} + k \binom{n/k}{2} <
        \gamma^2 n^2 + n^2 /k \leq
        \gamma^2 n^2 +
        n^2 (\log^*n)^{-\delta}.
      \end{align*}
    \item \textbf{Edges between $C_0$ and other $C_i$'s:}
      The total number is bounded by
      \begin{align*}
        \sum_{i=1}^{k} |C_0| |C_i| \leq |C_0| n \leq \gamma n^2.
      \end{align*}
    \item \textbf{Edges between irregular or sparse (i.e. density $< 8\gamma$) $C_i,C_j$'s:}
      The total number is bounded by
      \begin{align*}
        \gamma \binom{k}{2} (n/k)^2 +
        \sum_{1 \leq i<j\leq k: d(C_i,C_j) < 8\gamma} |E(C_i,C_j)|
        \leq 5 \gamma n^2.
      \end{align*}
  \end{enumerate}
  Summing up these and noting that
  $\gamma = (\log^* n)^{\delta}$ implies that $|F_1| + |F_2| \leq O(\gamma n^2)$.
  
  Moreover, since $F_3$ includes each good edge independently with probability $p = 10/\log n$, a simple Chernoff bound implies that w.h.p. $|F_3| \leq O(n^2/\log n) \ll \gamma n^2$, completing the proof.
\end{myproof}

The harder part of the proof, is to show that the sparse subgraph returned by \Cref{alg:matching-cover} is indeed a matching cover.  We continue with the following claim.

\begin{claim}\label{cl:goodpair-sample-nonempty}
  W.h.p., it holds for all $X\subseteq C_i, Y\subseteq C_j$ such that $(C_i, C_j)$ a good pair, $|X|\geq \gamma |C_i|$, and $|Y|\geq \gamma |C_j|$ that  $|F_3(X,Y)| > n^2 / \log^6 n$.
\end{claim}
\begin{myproof}
    Take a good pair $(C_i, C_j)$ and subsets $X\subseteq C_i, Y\subseteq C_j$ satisfying $|X|\geq \gamma |C_i|,|Y|\geq \gamma |C_j|$. Since $C_i,C_j$ is good, we know that $d(C_i,C_j) \geq 8\gamma$ and so
  $d(X,Y) \geq d(C_i,C_j) - \gamma \geq 7\gamma$ by the guarantee of the regularity lemma for good pairs.
  Thus, the number of edges between $X,Y$ in $E$ is
  \begin{align*}
    |E(X,Y)| \geq d(X,Y) |X| |Y|  \geq 7\gamma (\gamma |X|) (\gamma |Y|)  \geq 7 \gamma^3 (n/k)^2 \geq 7 (\log^* n)^{-3\delta} (n/\log n)^2  \geq 7 n^2 / \log^5 n.
  \end{align*}
  Since in obtaining $F$ we sample edges between $C_i,C_j$ with probability $p=10/\log n$, we get that
  $$
    \E|F(X, Y)| = p|E(X, Y)| \geq 7n^2/\log^6 n.
  $$
  Defining $\delta := \sqrt{6 n/\E|F(X, Y)|}$ and applying the Chernoff bound, we get that
  $$
    \Pr[|F(X, Y)| < (1-\delta)\E|F(X, Y)|] \leq \exp\left( - \frac{\delta^2 \E|F(X, Y)|}{2} \right) < e^{-3n}.
  $$
  Using the lower bound $7n^2/\log^6 n$ for $\E|F(X, Y)|$ and noting that $\delta$ is much smaller than, say, $1/2$ precisely because of this lower bound, we get that
  $$
    \Pr\left[|F(X, Y)| < \frac{n^2}{\log^6 n} \right] < e^{-3n}.
  $$
  Since the total choices of $C_i, C_j$ and subsets $X, Y$ is less than $\log^2 n \times 2^n \times 2^n < e^{2n}$, by a union bound we get that $|F(X, Y)| \geq n^2/\log^6 n$ for all $X$ and $Y$ and all $C_i$ and $C_j$ with probability at least $1-\frac{e^{2n}}{e^{3n}} \geq 1-2^{-n}$.
\end{myproof}

\begin{remark}\label{remark:over-sampling}
    For the purpose of this section, we will only use $|F_3(X, Y)| > 0$ instead of the much larger lower bound of \cref{cl:goodpair-sample-nonempty} for this set. This stronger guarantee will prove useful later in \cref{sec:dynamic} where we design our dynamic algorithm. We note that for the weaker guarantee of $|F_3(X, Y)| > 0$, it suffices to only sample $\Theta(n\log n)$ edges in \Cref{alg:matching-cover} instead of $O(n^2/\log n)$. Nonetheless, reducing the size of $F_3$ will not result in a sparser matching cover because $F_1$ and $F_2$ will already be in the order of $n^2/(\logstar n)^{\Theta(1)}$.
\end{remark}

Next, we prove the following lemma on {\em consolidating the support} of an arbitrary fractional matching so that each non-zero variable takes a sufficiently large value.

\begin{lemma}\label{lem:fractionalmatching}
    Let $\b{x}$ be any fractional matching (not necessarily in the matching polytope). For any $\epsilon \in (0, 1]$, there is a fractional matching $\b{y}$ such that all the following hold:
    \begin{enumerate}
        \item For any vertex $v$, $y_v \leq x_v$, where here $y_v := \sum_{e \ni v} y_e$ and $x_v := \sum_{e \ni v} x_e$.
        \item $\supp(\b{y}) \subseteq \supp(\b{x})$. That is, if $y_e > 0$ for some edge $e$, then  $x_e > 0$.
        \item For any edge $e$, either $y_e = 0$ or $y_e \geq \frac{\epsilon^3}{12\ln(1/\epsilon)}$.
        \item $|\b{y}| \geq |\b{x}| - 2\epsilon n$, where here $|\b{y}| := \sum_e y_e$ and $|\b{x}| := \sum_e x_e$.
    \end{enumerate}
\end{lemma}
\begin{proof}
    Let $\beta = 6 \ln(1/\epsilon)/\epsilon^3$; assume for ease of exposition that $\beta$ is an integer. For every edge $e \in E$ and every $i \in [\beta]$, draw a Bernoulli random variable $b^{i}_e$ that is 1 with probability $x_e$. Now for any edge $e$ we define
    $
        z_e := \beta^{-1} \sum_{i=1}^{\beta} b^{i}_e,
    $
    that is, $z_e$ is the fraction of the $\beta$ trials $b^1_e, \ldots, b^\beta_e$ that succeed. Moreover, for every vertex $v$ we let
    $
        z_v := \sum_{e \ni v} z_e,
    $ and construct $\b{y}$ as follows:
    \begin{flalign}\label{eq:hcl-290123}
        y_e = \mathbbm{1}[z_v < (1+\epsilon)x_v \text{ and } z_u < (1+\epsilon)x_u] \cdot z_e / (1+\epsilon) \qquad \qquad \text{for any edge $e$}.
    \end{flalign}
    
    First, observe that if $z_v \geq (1+\epsilon)x_v$ for a vertex $v$, then $y_v = 0$ by construction. Otherwise, we have $y_v \leq z_v/(1+\epsilon) < (1+\epsilon)x_v/(1+\epsilon) = x_v$. This proves the first property and also the fact that $\b{y}$ is a fractional matching.
    
    Second, note that if $x_e = 0$, then $z_e = 0$ and so $y_e = 0$. This immediately implies that any edge in the support of $\b{y}$ must also belong to the support of $\b{x}$, proving the second property.
    
    Third, observe from the definition of $z_e$ that either $z_e = 0$ or $z_e \geq \beta^{-1} = \epsilon^3/6\ln(1/\epsilon)$. This implies that either $y_e = 0$ or $y_e \geq \epsilon^3/6\ln(1/\epsilon)(1+\epsilon) \geq \epsilon^3 / 12\ln(1/\epsilon)$, proving the third property.
    
    We now turn to the fourth property. Take a vertex $v$. In the event that $z_v \geq (1+\epsilon)x_v$, we charge all of $z_v$ to $v$. That is, we define the charge $\phi(v)$ to be $z_v$ if $z_v \geq (1+\epsilon)x_v$ and $\phi(v)=0$ otherwise. We get that
    \begin{equation}\label{eq:mxc-2779018}
        \E|\b{y}| = \frac{1}{2} \sum_v \E[y_v] \geq \frac{1}{2(1+\epsilon)} \sum_v (\E[z_v - \phi(v)]) > (1-\epsilon) |\b{x}| - \frac{1}{2} \sum_v \E[\phi(v)],
    \end{equation}
    where the latter inequality follows from $\sum_v \E[z_v] = \sum_v x_v = 2|\b{x}|$. Next, we show that $\E[\phi(v)] \leq 2\epsilon$ for any vertex $v$. Observe that since $\E[\phi(v)] \leq \E[z_v] = x_v$, we are done if $x_v \leq 2\epsilon$. So let us assume that $x_v > 2\epsilon$. We have
    \begin{flalign}
        \nonumber \E[\phi(v)] &= \Pr[z_v \geq (1+\epsilon)x_v] \cdot \E[z_v \mid z_v \geq (1+\epsilon)x_v]\\
        \nonumber &\leq \E[z_v] - \Pr[|z_v - x_v| < \epsilon x_v] \cdot \E[z_v \mid |z_v - x_v| < \epsilon x_v]\\
        &\leq x_v - \Pr[|z_v - x_v| < \epsilon x_v] \cdot (1-\epsilon)x_v.\label{eq:mcg-9018377}
    \end{flalign}
    Now observe that $z_v \beta$ is a sum of $\deg(v) \cdot \beta$ independent (but not i.i.d.) Bernoulli random variables (one for each of the $\beta$ trials of each edge of $v$) with expected value $\E[z_v \beta] = x_v \beta$. Hence, by the Chernoff bound
    $$
        \Pr[|z_v\beta - x_v \beta| \geq \epsilon x_v \beta] \leq 2\exp\left( - \frac{\epsilon^2 x_v \beta}{3} \right) \stackrel{x > 2\epsilon, \beta = 6 \ln(1/\epsilon)/\epsilon^3}{\leq} 2\exp\left( - 4 \ln(1/\epsilon) \right) < 2\epsilon^4.
    $$
    Dividing through by $\beta$, this means that $\Pr[|z_v - x_v| \geq \epsilon x_v] < 2\epsilon^4$ and as a result $\Pr[|z_v - x_v| < \epsilon x_v] \geq 1-2\epsilon^4$. Plugging this into \cref{eq:mcg-9018377}, we get that
    $$
        \E[\phi(v)] \leq x_v - (1-2\epsilon^4)(1-\epsilon)x_v \leq 2\epsilon^2 x_v + \epsilon \leq 2\epsilon x_v \leq 2 \epsilon.
    $$
    Plugging $\E[\phi(v)] \leq 2\epsilon$ back into \cref{eq:mxc-2779018}, implies that
    $$
        \E|\b{y}| \geq (1-\epsilon)|\b{x}| - \epsilon n \geq |\b{x}| - 2\epsilon n.
    $$
    
    Finally, observe from our earlier discussion that the first three properties all hold with probability 1. The only place where we use the randomization of the construction of $\b{y}$ is for the fourth property where we showed $\E|\b{y}| \geq |\b{x}| - 2\epsilon$. This suffices for our purpose, since there must exist an outcome of $\b{y}$ with size as large as its expectation, while satisfying the other properties.
\end{proof}

We are now ready to prove that \Cref{alg:matching-cover} returns a $o(1)$-matching-cover w.h.p.

\begin{lemma}\label{lem:whpmatchingcover}
    The output of \Cref{alg:matching-cover} is, w.h.p., an $\alpha$-hitting set of $G$ for $\alpha = \Theta((\gamma \log(1/\gamma))^{1/3}) = (\log^*n)^{-\Omega(1)}$.
\end{lemma}
\begin{myproof}
    We show that if the event of \Cref{cl:goodpair-sample-nonempty} holds---which was shown to hold w.h.p. in the claim---then the output of \Cref{alg:matching-cover} is indeed an $\alpha$-hitting set.
    
    Let $A$ and $B$ be some arbitrary disjoint subsets of $V$ such that $|A|=|B|=\alpha n$ and let $M$ be a perfect matching of size $\alpha n$ between $A$ and $B$ in $G$. We have to show that there is at least one edge that goes from $A$ to $B$ in $H$. Observe that if there is any edge $e \in M$ that goes across a bad pair $C_i, C_j$ or if both endpoints of $e$ are in the same class $C_i$, then this edge $e$ will be our desired edge since \Cref{alg:matching-cover} will add $e$ to $F$. So let us assume that every edge in $M$ belongs to a good pair. Note that if we find a good pair $C_i, C_j$ where at least $\gamma |C_i| = \gamma |C_j|$ edges of $M$ go from $C_i$ to $C_j$, then we can immediately apply \cref{cl:goodpair-sample-nonempty} to get that there is one edge connecting two endpoints of matching $M$ in $F$. Unfortunately, such a pair may not always exist. Note that since there are $k$ classes,  each of the $\Theta(k^2)$ pairs might include only $\Theta(|M|/k^2) = \Theta((\gamma\log(1/\gamma)^{1/3} n /k^2)$ edges of $M$. This could be much smaller than $\gamma |C_i|$ which may be of size $\Omega(\gamma n/k)$. Additionally, we have to ensure that the edge that we find goes from $A$ to $B$, and that it does not connect $A$ to $A$ or $B$ to $B$. To achieve all of this, we give a more delicate argument that uses \cref{lem:fractionalmatching}. 
    
    First, let us slightly modify the matching $M$. Call a cluster $C_i$ {\em $A$-majority} if 
    $$
    |C_i \cap V(M) \cap A| > .75 |C_i \cap V(M)|,
    $$
    and similarly {\em $B$-majority} if
    $$
    |C_i \cap V(M) \cap B| > .75 |C_i \cap V(M)|.
    $$
    We construct a sub-matching $M'$ of $M$ by removing any edge $(u, v), u \in A, v \in B$ from it where $u$ is in a $B$-majority cluster or $v$ is in an $A$-majority cluster. Any cluster $C_i$ causes removal of at most $.25$ fraction of its matched vertices in $M$, hence the total number of removed vertices from $M$ is at most $.25|V(M)|$. Each such removed vertex may remove a unique edge in $M$. Thus, in total, the obtained matched $M'$ has size at least
    \begin{equation}\label{eq:bhgc-819388888}
        |M'| \geq |M| - .25|V(M)| = |M| - .5|M| = |M|/2 = \alpha n/2.
    \end{equation}

    Now for any $i, j \in [k]$ let $x_{ij}$ denote the fraction of the vertices of $C_i$ that are matched to $C_j$ in $M'$. Note that $x_{ij} = x_{ji}$ since the classes $C_1, \ldots, C_k$ all have equal sizes. Observe also that $\b{x}$ is a fractional matching of a complete graph on $k$ vertices: For any $i \in [k]$, $x_i := \sum_{j} x_{ij}$ equals the fraction of the vertices of $C_i$ that are matched by $M'$, and so $0 \leq x_i \leq 1$. Additionally, fractional matching $\b{x}$ satisfies the following:
    \begin{enumerate}[label=(X\arabic*)]
        \item If $x_{ij} \not= 0$, then $(C_i, C_j)$ is a good pair. This follows from our earlier assumption that all edges of $M'$ belong to good pairs.
        \item $|\b{x}| \geq \alpha k/2$. This holds because 
        $$
        |\b{x}| = \sum_{i,j \in [k]} x_{ij} = \sum_{i, j \in [k]} \frac{|M' \cap (C_i \times C_j)|}{|C_i|} = \frac{1}{|C_1|} \sum_{i, j \in [k]} |M' \cap (C_i \times C_j)| = \frac{|M'|}{|C_1|} \geq \frac{k|M'|}{n} \stackrel{\cref{eq:bhgc-819388888}}{\geq} \alpha k/2.
        $$
    \end{enumerate}
    
    Now we apply \cref{lem:fractionalmatching} on fractional matching $\b{x}$, for a parameter $\epsilon$ such that $\frac{\epsilon^3}{12\ln(1/\epsilon)} = 4\gamma$ which means $\epsilon = \Theta((\gamma \log(1/\gamma))^{1/3})$. This results in a fractional matching $\b{y}$ such that 
    \begin{enumerate}[label=(Y\arabic*)]
        \item For any $i \in [k]$, $y_i \leq x_i$.
        \item $\supp(\b{y}) \subseteq \supp(\b{x})$.
        \item For any $i, j \in [k]$ either $y_{ij} = 0$ or $y_{ij} \geq 4\gamma$,
        \item $|\b{y}| \geq |\b{x}| - \epsilon k.$
    \end{enumerate}
    Choosing the constant in the definition of $\alpha$ to be large enough such that $\alpha/2 > \epsilon$, we get from (X2) and (Y4) that $|\b{y}| > 0$. Hence, there must exist some $y_{ij} \not= 0$. This by (Y2) implies that $x_{ij} \not= 0$ and so $(C_i, C_j)$ must be a good pair. Additionally, $y_{ij} \not= 0$ implies by (Y3) that $y_i, y_j \geq y_{ij} \geq 4\gamma$ which by (Y1) also implies $x_i, x_j \geq 4\gamma$. From all of this, we get that there must be a good pair $(C_i, C_j)$ such that at least $4\gamma$ fraction of each of $C_i$ and $C_j$ is matched by $M'$, and there is an edge $(u, v) \in M'$ that goes from $u \in C_i$ to $v \in C_j$. Let us assume w.l.o.g. that $u \in A$, $v \in B$ (as we have not distinguished $C_i$, $C_j$ in any other way up to now). Our next claim is that
    \begin{enumerate}[label=(Z\arabic*)]
        \item There are subsets $X \subseteq C_i \cap A$, $Y \subseteq C_j \cap B$ such that $|X| \geq \gamma |C_i|, |Y| \geq \gamma |C_j|$.
        
        First, since $C_i$ has a vertex in $A$ that is matched in $M'$, then it cannot be $B$-majority. Hence, $|C_i \cap V(M) \cap A| \geq |C_i \cap V(M)|/4 \geq |C_i \cap V(M')|/4 \geq \cdot 4\gamma |C_i|/4 = \gamma |C_i|$, where the third inequality follows from our earlier discussion that at least $4\gamma$ fraction of vertices in $C_i$ are matched by $M'$. Similarly, since $C_j$ has a vertex in $B$ that is matched in $M'$, it cannot be $A$-majority and thus $|C_j \cap V(M) \cap B| \geq |C_j \cap V(M)|/4 \geq |C_j \cap V(M')|/4 \geq \cdot 4\gamma |C_j|/4 = \gamma |C_j|$.
    \end{enumerate}

    Now applying \cref{cl:goodpair-sample-nonempty} on these subsets $X$ and $Y$ of $C_i$ and $C_j$ in (Z1), proves that subsample $F$ must include at least one edge between them, and so $F$ is an $\alpha$-hitting set.
\end{myproof}

\subsection{Second Step: From Hitting Set to Matching Cover}\label{sec:hitting-to-matching-cover}

We now prove that any subgraph satisfying
the hitting set requirement (\cref{def:hitting-set})
is also a matching cover (\cref{def:matching-cover}).
This will follow from Hall's theorem (Proposition~\ref{prop:hall}),
and the proof similar to that of Lemma 9.3 in~\cite{GoelKK12}.

\begin{lemma}[From Hitting Set to Matching Cover]
\label{lem:hitcover}
Let $G=(V,E)$ be any graph that is not necessarily bipartite.
Then any subgraph $H$ of $G$ that is an $\alpha$-hitting
set is also an $\alpha$-matching cover of $G$.
\end{lemma}
\begin{proof}
    Let $H$ be a subgraph of $G$ that is an $\alpha$-hitting set.
    Consider any disjoint subsets $A,B\subset V$
    with a maximum matching size of $\mu_G(A,B)$ in $G$.
    Let $P\subseteq A, Q\subseteq B$ be subsets
    with $|P|=|Q|=\mu_G(A,B)$ such that
    there exists a perfect matching $M^*$ from $P$ to $Q$.
    We show that there must exist
    a matching from $P$ to $Q$ in $H$ of size $|P| - \alpha n$,
    i.e. $\mu_H(P,Q)\geq |P| - \alpha n$.
    
    Suppose for the sake of contradiction that
    $\mu_H(P,Q) < |P| - \alpha n$.
    Then by applying the extended Hall's theorem (Proposition~\ref{prop:hall})
    to the bipartite subgraph of $H$ induced by
    the bipartition $(P,Q)$,
    there exists $P'\subset P$ such that
    $|P'| - |N_H(P')\intersect Q| > \alpha n$.
    Now consider the edges of $M^*$ that are
    incident on $P'$ but
    {\em not} incident
    on $N_H(P')$, which themselves
    form a matching (call it $M^{**}$)
    of size $> \alpha n$.
    Now by the fact that $H$ is an $\alpha$-hitting set of
    $G$, there must be an edge in $H$
    connecting $V(M^{**})\intersect P$ to
    $V(M^{**})\intersect Q$.
    On the other hand,
    by the definition of $M^{**}$,
    we have $(V(M^{**}) \intersect P)\subseteq P'$
    but $V(M^{**})\intersect N_H(P') = \emptyset$,
    leading to a contradiction.
\end{proof}

We are now ready to prove \cref{thm:matching-cover-polytime}.

\begin{proof}[Proof of \cref{thm:matching-cover-polytime}]
    The output of \Cref{alg:matching-cover} being an $o(1)$-hitting set was proved in \cref{lem:whpmatchingcover}.
    By~\Cref{lem:hitcover},
    the output subgraph is an $o(1)$-matching cover.
    This matching cover having at most $O(\gamma n^2) = n^2/(\log^* n)^{\Omega(1)}$ edges was proved in \cref{cl:fewedgespicked}. Finally, the running time follows from the algorithm of \cref{thm:regularityalgo} for finding the regularity decomposition, and the fact that $Q(t,1/\gamma) \leq \log n$ in \Cref{alg:matching-cover}.
\end{proof}

\section{A Fully Dynamic Algorithm via Matching Covers}\label{sec:dynamic}

In this section, we show that the matching cover of \cref{sec:matcover} can be used to prove the following result in the fully dynamic model.

\begin{theorem}\label{thm:dynamicalgo}
There is a randomized, fully dynamic algorithm
that maintains with high probability a $(1-o(1))$-approximate matching
under (possibly adversarial) edge updates.
The algorithm has initialization time
$O(n^{\omega} \log n)$
and worst-case update time $n/(\log^* n)^{\Omega(1)}$.
\end{theorem}

We start by giving an overview
of our algorithm.
We first describe a strategy that enables us
to maintain an approximate matching with additive error
$o(n)$, and latter explain how to make the approximation guarantee
multiplicative.
We re-compute an $o(1)$-matching cover of the current
graph every $\Theta(n^{\omega-1}\log^2 n)$ updates,
and then pretend as if the matching cover is the entire graph,
and use
the $O(\sqrt{m})$ update-time algorithm of \citet*{GuptaP13}, stated below as \cref{prop:gupta-peng},
to maintain a nearly optimal matching.

First, it is easy to see that the amortized update time
of this strategy is $o(n)$, as
the computation time $O(n^{\omega}\log n)$ of the matching cover
gets
amortized over $\Theta(n^{\omega-1}\log^2 n)$ updates to $o(n)$,
and the number of edges in the matching cover
is $o(n^2)$.
Then to argue the correctness,
we have to show that the matching cover found
by our offline algorithm
is {\em robust to edge updates} -
that is, it remains an $o(1)$-matching cover
throughout the following $\Theta(n^{\omega-1}\log^2 n)$ updates.
This is indeed a feature of our offline algorithm:
in particular,
the number of edges
between each pair of large enough $X\subseteq C_i, Y\subseteq C_j$
for dense, regular $C_i,C_j$ pairs
is $\tilde{\Omega}(n^2)$,
which means that the hitting set property
will be preserved as long as
$\ll n^2$ edges have been deleted,
and as a result the subgraph obtained by our algorithm
remains an $o(1)$-matching cover throughout
the following $\Theta(n^{\omega-1}\log^2 n) \ll n^2$ updates, as desired.

To turn the additive approximation guarantee to a multiplicative one,
we will deal with ``sparse'' and ``dense'' regimes separately.
Specifically,
when the number of edges 
is at most $n^2/(\log^* n)^{\Omega(1)}$, we simply use the Gupta-Peng algorithm to maintain a $(1-\epsilon)$-approximation in $n/(\log^* n)^{\Omega(1)}$ update time. On the other hand, when the graph is dense, we first use the matching cover of \Cref{thm:matching-cover-polytime} to sparsify the graph while preserving its maximum matching, then run \cref{prop:gupta-peng} on this sparse graph to maintain a $(1-\epsilon)$-approximate matching of it in $n/(\log^* n)^{\Omega(1)}$ update-time. 
We also set up a ``buffer zone'' in the thresholds for switching between the two algorithms so that we do not pay the switching overhead too often.

We now present our algorithm.
We  first show an algorithm with initialization time $O(n^{\omega} \log n)$
and {\em amortized} update time $n/(\log^* n)^{\Omega(1)}$,
and then discuss how to make the update time worst-case. To present our algorithm, we need the following result by~\citet*{GuptaP13}. 

\begin{proposition}[\cite{GuptaP13}]\label{prop:gupta-peng}
There is a deterministic, fully dynamic algorithm
for maintaining a $(1-\eps)$-approximate matching
with initialization time $O(m_0\eps^{-1})$ and
worst-case update time $O(\sqrt{m} \eps^{-2})$,
where $m_0$ is the number of edges in the initial graph,
and $m$ is the maximum number of edges
in the graph throughout the updates.
\end{proposition}

Our algorithm is formally as follows. 
\begin{Algorithm}\label{alg:dynamic}
  \begin{center}
		A fully dynamic algorithm for \cref{thm:dynamicalgo}.
	\end{center}
  \medskip
  
  \textbf{Input:} An $n$-vertex fully dynamic graph $G$ subject to edge insertions and deletions.
  
  \smallskip
  
  \textbf{Output:} A (dynamically changing) $(1-\eps)$-approximate maximum of $G$.
  
  \smallskip

  \textbf{Parameters:} We set $t$, $\gamma$, $\delta$ as in \Cref{alg:matching-cover}, and set $\eps \gets 10 (\log^*n)^{-\delta/64}$.

  \smallskip\smallskip

\textbf{Sparse regime:}

\begin{enumerate}[itemsep=0pt,topsep=5pt]
    \item Whenever the number of edges in $G$ exceeds
      $n^2/(\log^* n)^{\delta/8}$, switch to the dense regime.
    \item Use \cref{prop:gupta-peng} on the whole graph $G$ to maintain a $(1-\epsilon)$-approximation.
\end{enumerate}

\textbf{Dense regime:}

\begin{enumerate}[itemsep=0pt,topsep=5pt]
    \item Whenever the number of edges in graph $G$ falls below
    $n^2/(2(\log^* n)^{\delta/8})$,
    restart the algorithm of \cref{prop:gupta-peng} on the whole graph $G$ to maintain a $(1-\eps)$-approximate matching of it, and switch to the sparse regime.
    \item Do the following every $n^{\omega-1}\log^2 n$ updates
      (including before the first update):
    \begin{enumerate}
    \item Use \Cref{alg:matching-cover} to construct an $o(1)$-matching cover $F$ of the graph $G$ in $O(n^\omega \log n)$ time (see \cref{thm:matching-cover-polytime}).
    \item Restart the algorithm of \cref{prop:gupta-peng} for maintaining a $(1-\eps)$-approximation of subgraph $F$.\label{line:2c}
    \end{enumerate}
    \item Upon insertion of an edge
      $e$, let $F\gets F\union \{e\}$ and trigger an edge insertion to the algorithm of \cref{prop:gupta-peng} we use on $F$.
    \item Upon deletion of an edge $e$,
      if $e\in F$, let $F\gets F - \{e\}$ and trigger an edge deletion to the algorithm of \cref{prop:gupta-peng} we use on $F$;
      otherwise ignore the deletion.
\end{enumerate}
\end{Algorithm}

We now turn to analyze \Cref{alg:dynamic}. First, we prove that it has our desired update-time via amortization. As discuss, we will later show how the algorithm can be deamortized.

\begin{claim}\label{clm:dynamic-runtime}
  The amortized update-time of \Cref{alg:dynamic} is $n/(\log^* n)^{\Omega(1)}$.
\end{claim}
\begin{myproof}
    When the algorithm is in the sparse regime, there are at most $m \leq  n^2/(\logstar n)^{\delta/8}$ edges in it. Since in this case we use \cref{prop:gupta-peng} on the whole graph, the update-time is at most 
    $$
    O(\sqrt{m}/\epsilon^2) = O\left(\sqrt{n^2/(\logstar n)^{\delta/8}} / \left(10(\logstar n)^{-\delta/64} \right)^2\right) = O(n/(\logstar n)^{\delta/32}) = n/(\logstar n)^{\Omega(1)},
    $$
    where the last equality holds since we set parameter $\delta$ to be a constant in \Cref{alg:matching-cover}.
    Also, the initialization time is 
    $$
    O(m_0 / \eps) = 
    O\left(
    n^2/
    (
    2(\log^*n)^{\delta/8})
    / (10(\logstar n)^{-\delta/64})
    \right)
    = O(n^2 / (\log^*n)^{7\delta/64}).
    $$
    Due to the threshold gap for switching between the sparse
    and dense cases, the initialization only happens
    every $n^2/(2(\log^*n)^{\delta/8})$ updates.
    This coupled with that $\delta$ is a constant means that
    the running time of the initialization gets
    amortized to $(\log^*n)^{O(1)} \ll n/(\log^*n)^{\Omega(1)}$.
    
    For the dense regime, we analyze the amortized cost of running \Cref{alg:matching-cover} and the algorithm of \Cref{prop:gupta-peng} on $F$ separately.
    
    Observe that we call \Cref{alg:matching-cover} either if it we are in the dense regime and there has been $n^{\omega-1}\log^2 n$ updates since we last called it, or if we switch from the sparse regime to the dense regime. Once again because of the threshold gap for switching from sparse to dense vs. from dense to sparse regimes, the latter type of calls to \Cref{alg:matching-cover} only happen every $n^2/(2(\logstar n)^{\delta/8}) \gg n^{\omega-1}\log^2 n$ (since $\omega < 2.373$) updates. As such, since \Cref{alg:matching-cover} takes $O(n^\omega \log n)$ time by \Cref{thm:matching-cover-polytime}, the overall amortized cost of running \Cref{alg:matching-cover} is only $O\left(n^\omega \log n / (n^{\omega-1} \log^2 n) \right) = O(n/\log n)$.

    Next, note that immediately after we run \Cref{alg:matching-cover}, the set $F$ only includes $O(\gamma n^2) = O(n^2/(\logstar n)^\delta)$ edges by \cref{cl:fewedgespicked}. Within the next $n^{\omega-1}\log^2 n$ updates until we call \Cref{alg:matching-cover} again, we may add up to $n^{\omega-1}\log^2 n$ other edges to $F$. Therefore, at any point $F$ will include at most $O(n^2/(\logstar n)^\delta + n^{\omega-1}\log^2 n) = O(n^2/(\logstar n)^\delta)$ edges. This means that the update-time of \cref{prop:gupta-peng} for the dense regime is at most
    $$
    O(\sqrt{|F|}/\epsilon^2) = O\left(\sqrt{n^2/(\logstar n)^{\delta}} / \left(10(\logstar n)^{-\delta/64} \right)^2\right) = O(n/(\logstar n)^{15\delta/32}) = n/(\logstar n)^{\Omega(1)}.
    $$
    
    Adding up all the computation costs above, we get that the algorithm has an overall amortized update time of $n/(\logstar n)^{\Omega(1)}$.
\end{myproof}

Next, we prove that \Cref{alg:dynamic} maintains a $(1-o(1))$-approximate matching w.h.p.

\begin{claim}\label{clm:dynamic-correct}
  At any point, the output of \Cref{alg:dynamic} is w.h.p. a $(1-o(1))$-approximate maximum matching of $G$. This holds, in particular, against an adaptive adversary that is aware of both the output and the state of the algorithm.
\end{claim}
\begin{myproof}
    For the sparse regime, this directly
follows from the correctness of \cref{prop:gupta-peng} since we run it on the entire graph $G$.
We thus focus on the dense regime.

First, note that in the dense regime there are at least $m \geq n^2/(2(\logstar n)^{\delta/8})$ edges in the graph. Observe that any $n$-vertex $m$-edge graph has a matching of size at least $m/(2n-1)$: iteratively pick an arbitrary free edge, add it to the matching, and remove its endpoints from the graph; each step only removes at most $(2n-1)$ edges, thus the matching must have size at least $m/(2n-1)$. From this, we get that whenever the algorithm is in the dense regime, there is a matching of size at least $\mu(G) \geq n/(4(\logstar n)^{\delta/8})$ in it.

Next, we claim that at any point in the dense regime, $F$ is an $\alpha$-matching cover
of $G$ (\cref{def:matching-cover}), where as defined in \Cref{lem:whpmatchingcover},
$$
\alpha = \Theta((\gamma \log(1/\gamma))^{1/3}) =  \Theta((\logstar n)^{-\delta} \log((\logstar n)^{\delta}) )^{1/3}) = O((\logstar n)^{-\delta/4}).
$$
By \cref{lem:hitcover}, to show this,
it suffices to show that $F$ is an
$\alpha$-hitting set of $G$ (\cref{def:hitting-set})
at any point in the dense regime.
To see this, observe that immediately after we call \Cref{alg:matching-cover},  $F$ must be an $\alpha$-hitting set of $G$ simply by the guarantee of \Cref{thm:matching-cover-polytime}. However, for the next $n^{\omega-1}\log^2 n$ updates until we re-run \Cref{alg:matching-cover}, both the graph $G$ and $F$ change due to the updates to the graph. Observe that edge insertions cause no problem since any edge added will be added to $F$ as well. But edge deletions may cause a problem. In particular, recall that we subsample $o(1)$ fraction of edges of the good pairs in \Cref{alg:matching-cover}, and if they are all removed then we no longer have an $\alpha$-hitting set. Indeed, given that the adaptive adversary is aware of this sampled subset, he can attempt to remove these edges one by one. The crucial observation, here, is that right after we call \Cref{alg:matching-cover},  \Cref{cl:goodpair-sample-nonempty} guarantees that there are w.h.p. at least $|F_3(X, Y)| \geq n^2/\log^6 n$ subsampled edges between any two large enough subsets $X \subseteq C_i, Y \subseteq C_j$ of any good pair $(C_i, C_j)$. On the other hand, our guarantee of \Cref{thm:matching-cover-polytime} that $F$ is an $\alpha$-hitting set only requires $|F_3(X, Y)| > 0$ (see \Cref{remark:over-sampling}). As a result, even if the adversary attempts to remove edges of $F_3$ one by one within the next $n^{\omega-1}\log^2 n \ll n^2/\log^6 n$ updates, $F_3(X, Y)$ will remain non-empty and so $F$ remains an $\alpha$-hitting set.

Moreover, since $F$ is an $\alpha$-matching cover of $G$, we get from \cref{def:matching-cover}, taking $M^\star$ to be an arbitrary maximum matching of $G$ and taking sets $A$ and $B$ to each include one endpoint of each edge in $M^\star$ arbitrarily, we get that 
$$
    \mu(F) \geq \mu(F[A, B]) \geq \mu(G[A, B]) - \alpha n = |M^\star| - \alpha n = \mu(G) - \alpha n.
$$

Putting together the bounds above, we get that at any point during the updates in the dense regime, $F$ includes a matching of size at least 
$$
\mu(F) \geq \mu(G) - \alpha n = \mu(G) - O(n/(\logstar n)^{\delta/4}) \geq (1-o(1))\mu(G),
$$
where the last equality holds since $\mu(G) \geq \Omega(n/(\logstar n)^{\delta/8})$ as discussed above. 
Running the algorithm of \cref{prop:gupta-peng} on top of this, we maintain a $(1-\epsilon)(1-o(1))\mu(G) = (1-o(1))\mu(G)$ size matching overall.
\end{myproof}

\def\gg{G_{\mathrm{good}}}
\def\mg{M_{\mathrm{good}}}

\subsection{From Amortized to Worst-case Update Time}

We now show how to make the update time worst-case
without blowing up the update time by more than a constant factor.

In the sparse regime,
our update time is already worst-case as guaranteed
by \cref{prop:gupta-peng}.
We then consider how to make the update time worst-case in the dense regime.

First note that the running time of our algorithm for the dense regime does not
depend on the number of edges.
Thus we could always run the dense regime algorithm in the background,
and only adopt its solution when the graph is dense.
To make the update time worst-case,
we distribute the computation of
a matching cover evenly over
the following $n^{\omega-1}\log^2 n$ updates,
and then distribute the initialization of the data structure with edges $F$
evenly across the $n^{\omega-1}\log^2 n$ updates after.
Of course, at this point the data structure will be falling behind
by $2n^{\omega-1}\log^2 n$ updates. We will then catch it up
in the following $n^{\omega-1}\log^2 n$ updates, by feeding
it $3$ updates per update.
At any point, we will always use the data structure that is up-to-date,
and discard it as soon as a new data structure has become up-to-date.
This way any data structure only goes through $O(n^{\omega-1}\log^2n) $ updates.

Note that we still have to address the switch from the dense regime to sparse regime,
where we restart \cref{prop:gupta-peng} algorithm by initiating the data structure with the edges
of the entire current graph.
To this end, we actually also always run the \cref{prop:gupta-peng} algorithm in the background,
and only use its solution if needed. But since the running time of the algorithm depends
on the number of edges, we will make sure that the number of edges in the data structure is always
bounded by $n^2 (\log^*n)^{-\delta/8}$.
Specifically, upon the insertion of an edge $e$, if the number of edges in the data structure is
already $n^2 (\log^*n)^{-\delta/8}$, we do not insert $e$ into the data structure,
but store $e$ in a linked list $L$.
Upon the deletion of an edge $e$,
if it is currently in the data structure,
we delete it from the data structure;
otherwise we delete it from $L$.
Moreover,
whenever the number of edges in the data structure becomes strictly
less than $n^2 (\log^*n)^{-\delta/8}$ after a deletion,
we immediately insert an edge in $L$ (if any) into the data structure,
and delete that edge from $L$.
This way, in the sparse regime, it is guaranteed that all edges are in the data structure.

\Cref{thm:dynamicalgo} now follows from~\Cref{clm:dynamic-correct} for the correctness of~\Cref{alg:dynamic} and~\Cref{clm:dynamic-runtime} and the discussions in this subsection for its runtime.

\newcommand{\MatchingCover}{\ensuremath{\textnormal{\textsf{Matching-Cover}}}\xspace}

\section{Single-Pass Streaming Algorithms} 

We prove~\Cref{res:stream1} and~\Cref{res:stream2} in this section. Both algorithms rely on using matching covers iteratively in the same way and differ primarily on how they compute matching covers and some additional steps. Because of this, we first present and prove a generic result that uses matching covers in a blackbox way to obtain a streaming algorithm for finding matching covers and then extend it separately to obtain for~\Cref{res:stream1} and~\Cref{res:stream2}.

\subsection{A Streaming Algorithm for Matching Covers} 

We present an algorithm that computes the matching cover of a graph presented in a stream by iteratively computing matching covers of smaller subsets of the stream without losing ``much'' on the quality of the final matching cover.
For technical reasons that will become clear later, we need this algorithm
to work for multi-graphs as well. 

\begin{proposition}\label{prop:streaming-mc}
	For any integer $k \geq 1$ and any  $\alpha \in (0,1/10)$, there exists a single-pass streaming algorithm that computes an $\alpha$-matching cover of  $n$-vertex multi-graphs with at most $m$ edges in space 
	\[
		O\Paren{\frac{m}{k} \cdot \log\paren{\frac{n^2 \cdot k}{m}} + \MC{n}{\alpha/{2k}} \cdot \log{\paren{\frac{n^2}{\MC{n}{{\alpha}/{2k}}}}} \cdot \log{k}};
	\] 
	here, we assume we are given a subroutine $\MatchingCover$ that given adjacency matrix access to any $n$-vertex graph with $m/k$ edges, can compute an $(\alpha/2k)$-matching cover with at most $\MC{n}{\alpha/2k}$ edges using $O((m/k) \cdot \log{({n^2 \cdot k}/{m})})$ space. The streaming algorithm requires calling $\MatchingCover$ $O(k)$ times and is deterministic as long as the $\MatchingCover$ subroutine is deterministic.
\end{proposition}

The algorithm in~\Cref{prop:streaming-mc} is based on a novel use and modification of the widely used ``Merge and Reduce'' technique in the streaming literature (used previously e.g., for quantile estimation~\cite{MankuRL99,KarninLL16} or cut/spectral sparsifiers~\cite{McGregor14}). We give a high level overview of the algorithm here and present the formal description in~\Cref{alg:rs-matching}. 

The algorithm maintains $t:=O(\log{k})$ different buffers $B_1,\ldots,B_t$ of edges throughout the stream (all these buffers store their edges using the succinct dynamic dictionary of~\Cref{prop:succinct-dict} to save space). Buffer $B_1$ simply starts reading edges from the stream until it collects $m/k$ edges; it will then use the (offline) subroutine $\MatchingCover$ over these edges with parameter $\alpha' = \alpha/2k$ to obtain an $\alpha'$-matching cover of the subgraph of input on edges in $B_1$. Edges of this matching cover are then inserted to buffer $B_2$ and we empty buffer $B_1$, which will continue reading edges from the stream again. In the mean time, whenever buffer $B_2$ gets ``full'', this time meaning that it receives twice as 
many edges as $\MC{n}{\alpha'}$, we compute another $\alpha'$-matching cover using $\MatchingCover$, this time over the edges in $B_2$, pass them to buffer $B_3$, and empty $B_2$ which continues receiving edges from buffer $B_1$. This process is done the same way across all buffers until all edges of the stream have passed (we prove buffer $B_t$ never gets full so not having a buffer $B_{t+1}$ is not a problem). At the end, we argue that the edges that are remained across all buffers $B_1,\ldots,B_t$ at the end of the stream form an $\alpha$-matching cover of the input. 

The analysis of the algorithm involves showing that: $(i)$ fewer and fewer edges find their way to higher-indexed buffers, $(ii)$ the repeated application of $\MatchingCover$ does not blow up the approximation guarantee by too much, and $(iii)$ all this can be implemented in a relatively small space. We now present the formal algorithm and its analysis. 


\begin{Algorithm}\label{alg:rs-matching}
	\begin{center}
		An algorithm for~\Cref{prop:streaming-mc}. 
	\end{center}
	\medskip
	
	\textbf{Input:} A multi-graph $G=(V,E)$ in the stream with $n$ edges and at most $m$ edges. We are also given integer $k \geq 1$ and approximation parameter $\alpha > 0$, and access to the (offline) subroutine $\MatchingCover$ as specified in~\Cref{prop:streaming-mc}. 
	
	\smallskip
	
	\textbf{Output:} An $\alpha$-matching cover of $G$. 
	
	\smallskip
	
	\textbf{Parameters:} We set $t:= (\log{k}+2)$ and $\alpha' := \alpha/2k$. 
	
	\smallskip
	
	\begin{enumerate}[label=$(\roman*)$]
		\item Maintain the following \textbf{buffers} of edges $B_1,\ldots,B_t$ using succinct dynamic dictionary of~\Cref{prop:succinct-dict} (we specify the details in~\Cref{lem:rs-space}): 
		
		\begin{enumerate}
			\item Buffer $B_1$: add any arriving edge $(u,v)$ arrives in the stream to $B_1$. Once size of $B_1$ reaches $m/k$, 
			run $\MatchingCover$ to find an $\alpha'$-matching-cover of the subgraph $(V,B_1)$ of $G$ and add all those
			edges to $B_2$. Restart $B_1$ by deleting all its current edges.  
			
			\item Buffers $B_i$ for $i > 1$: once size of $B_i$ reaches $2 \cdot \MC{n}{\alpha'}$, run $\MatchingCover$ to find an $\alpha'$-matching-cover of the subgraph $(V,B_i)$ of $G$ and 
			add all those edges to $B_{i+1}$\footnote{We will show in~\Cref{clm:k-i} that this step never happens for buffer $B_t$, namely, it never gets ``full'', and thus the algorithm is well-defined even though there is no buffer $B_{t+1}$.}. Restart $B_i$ by deleting all its current edges. 
				\end{enumerate}
		
		\item Return $(B_1 \cup \ldots \cup B_t)$ at the end of the stream. 
	\end{enumerate}
\end{Algorithm}

We start by analyzing the space complexity of~\Cref{alg:rs-matching}. 

\begin{lemma}\label{lem:rs-space}
	\Cref{alg:rs-matching} can be implemented in space of
	\[
		O\paren{\frac{m}{k} \cdot \log\paren{\frac{n^2 \cdot k}{m}} + t \cdot \MC{n}{\alpha'} \cdot \log\paren{\frac{n^2}{\MC{n}{\alpha'}}}}.
	\]
\end{lemma}
\begin{proof}
	At any point in the algorithm, $B_1$ contains $s_1=O(m/k)$ edges and each $B_i$ for $i > 1$ contains $s_i = O(\MC{n}{\alpha'})$ edges.  
	We can maintain each buffer $B_i$ for $i \in [t]$ using a dedicated data structure $\DS_i$ of~\Cref{prop:succinct-dict} for the parameter $s_i$ over the universe of all pairs of vertices: 
	\begin{itemize}
	\item To add any edge $(u,v)$ to a buffer $B_i$, we first check if $(u,v)$ belongs to $B_i$ via $\DS_i.\,\member(u,v)$, and if not use $\DS_i.\,\add(u,v)$ to add the edge to $B_i$. 
	\item To delete all edges from $B_i$, we simply erase $\DS_i$ and start it from scratch. 
	\item $\DS_i$ now gives us an adjacency matrix access to the subgraph $(V,B_i)$ by checking $\DS_i.\,\member(u,v)$ for finding if $(u,v)$ is an edge in the subgraph. 
	\end{itemize}
	Notice that even though $G$ can be a multi-graph, each individual $(V,B_i)$ is a simple graph. 

	By~\Cref{prop:succinct-dict}, 
	the space needed for storing $\DS_1$ and each $\DS_i$ for $i > 1$ is, respectively, 
	\begin{align*}
		&(1+o(1)) \cdot \log{{{n^2}\choose{m/k}}} = O(\frac{m}{k} \cdot \log\paren{\frac{n^2 \cdot k}{m}}), \\
		&(1+o(1)) \cdot \log{{{n^2}\choose{\MC{n}{\alpha'}}}} = O(\MC{n}{\alpha'} \cdot \log{(\frac{n^2}{\MC{n}{\alpha'}})}).
	\end{align*}
	Given there are $t-1$ buckets of the latter type, this bounds the space needed to store all the buffers in the algorithm as required in the lemma statement. 
	
	Finally, to implement each run of the subroutine $\MatchingCover$, since we have stored $\DS_i$ for the buffer $B_i$, we can provide an adjacency matrix access to $B_i$ for $\MatchingCover$ (as required by the \Cref{prop:streaming-mc} statement), by simply checking $\DS_i.\,\member(u,v)$ for any query $(u,v)$ to the adjacency matrix. As $\MatchingCover$ is promised to use $O((m/k) \cdot \log{(n^2 \cdot k}/{m}))$ space with this access, we get that the final bound on the space complexity of \Cref{alg:rs-matching}. \end{proof}
\noindent
We now prove the correctness of~\Cref{alg:rs-matching}. To do so, we need the following definitions: 
\begin{itemize}
	\item Let $H^1_1,\ldots,H^1_{k_1}$ denote the $k_1$ separate matching covers constructed by the algorithm over the edges of buffer $B_1$, one for each time that we restart $B_1$. 
	Let $G^2 := H^1_1 \cup \ldots \cup H^1_{k_1}$ denote the graph that is sent to buffer $B_2$ throughout the algorithm (for notational convenience, we also define $G^1=G$ as the input graph, 
	namely, the graph that is sent to buffer $B^1$). 
	
	\item For any $i \in [2: t-1]$, similarly, let $H^i_1,\ldots,H^i_{k_i}$ denote the $k_i$ separate matching covers constructed by the algorithm over the edges of buffer $B_i$. Let 
	$G^{i+1} := H^i_1 \cup \ldots \cup H^i_{k_i}$ denote the graph that is sent to buffer $B_{i+1}$ throughout the algorithm.
\end{itemize}

We prove that the number of subgraphs at buffer $B_i$ drops by a factor of $2^{i}$ compared to $B_1$. 

\begin{claim}\label{clm:k-i}
	For any $i \in [t-1]$, $k_i \leq k/2^{i-1}$ and $k_t = 0$ meaning that bucket $B_{t}$ never generates a matching cover (namely, it never gets full). 
\end{claim}
\begin{proof}
	We prove $k_i \leq k/2^{i-1}$ inductively. For the base case, since we restart buffer $B_1$ after each $m/k$ edges in the stream and there are at most $m$ edges in 
	the stream, we have  $k_1 \leq k$. For $i > 1$, the algorithm creates an $\alpha'$-matching-cover $H^i_j$ whenever bucket $B_i$ gets full, which happens only when it collects $2 \cdot \MC{n}{\alpha'}$ edges. 
	Moreover, the total number of edges ever sent to the bucket $B_i$ is $\card{E(G^{i})}$ by the definition of the subgraph $G^{i}$. Thus, 
	\begin{align*}
		k_i &\leq \frac{1}{2 \cdot \MC{n}{\alpha'}} \cdot \card{E(G^i)} \tag{as $k_i$ is equal to the number of times $B_i$ gets full} \\
		&\leq  \frac{1}{2 \cdot \MC{n}{\alpha'}} \cdot \sum_{j=1}^{k_{i-1}} \card{E(H^{i-1}_j)} \tag{as $G^{i}$ is a union of $k_{i-1}$ matching-covers $H^{i-1}_j$} \\
		&\leq  \frac{1}{2 \cdot \MC{n}{\alpha'}} \cdot k_{i-1} \cdot \MC{n}{\alpha'} \tag{by the guarantee of $\MatchingCover$, $\card{E(H^{i-1}_j)} \leq  \MC{n}{\alpha'}$} \\
		&\leq \frac{k}{2^{i-1}},
	\end{align*}
	where the last step is by the induction hypothesis for $i-1$. This proves the first part of the claim. 
	
	We now have that $k_{t-1} \leq k/2^{t-2} = k/2^{\log{k}}= 1$.  Thus, only one matching-cover is ever sent to $B_t$ and so $B_t$ receives at most $\MC{n}{\alpha'}$ edges and  never gets full. 
\end{proof}

The following lemma captures the loss on the size of maximum matching that the algorithm maintains from one buffer to the next one. In other words, the cost we have
to pay for introduction of each level of buffers. 
\begin{lemma}\label{lem:rs-approx}
	For any $i \in [t-1]$ and any disjoint subsets of vertices $X,Y \subseteq V$, 
	\[
	\mu\paren{(G^{i+1} \cup B^f_{i} \cup \ldots \cup B^f_1)[X,Y])} \geq \mu\paren{(G^{i} \cup B^f_{i-1} \cup \ldots \cup B^f_1)[X,Y]} - k_i \cdot \alpha' \cdot n. 
	\]
	 where  $B^f_j$ for $j \in [t]$ is the \underline{final} content of the buffer at the end of the stream. 
\end{lemma}
\begin{proof}
	Fix any $i \in [t-1]$ and a maximum matching $M^*_{i}$ of $(G^{i} \cup B^f_{i-1} \cup \ldots \cup B^f_1)[X,Y]$. We construct a matching $M_{i+1}$ in $(G^{i+1} \cup B^f_i \cup \ldots \cup B^f_1)[X,Y]$ 
	such that $\card{M_{i+1}} \geq \card{M^*_i} - k_i \cdot \delta \cdot n$. This will then immediately implies the lemma. To continue we need some more definition. 
	
	For any $H^i_j$ for $j \in [k_i]$, let $B^{i}_j$ denote the content of buffer $B_i$ when the algorithm creates $H^i_j$. This way, 
	$H^i_j$ is a matching-cover of $(V,B^i_j)$. Moreover, $B^i_1,\ldots,B^i_{k_i}$ together with $B^f_i$ partition all the edges that are ever sent to buffer $B_i$, namely, the graph $G^i$.  
	These edges are also further disjoint from $B^f_{i-1},\ldots,B^f_1$ since the latter set of edges were never sent to buffer $B_i$. We can partition the edges of $M^*_i$  
	between these sets and along the way define our matching $M_{i+1}$ as well: 
	\begin{itemize}
		\item For any $j \in [k_i]$, let $M^*_{i,j} := M^*_{i} \cap B^i_j$ and $M_{i,j}$ be the maximum matching in $H^i_j$ between $X(M^*_{i,j})$ and $Y(M^*_{i,j})$.  
		\item For any $i' \in [i]$, let $M^{*,f}_{i'} := M^*_i \cap B^f_i$ and $M^f_{i'} := M^{*,f}_{i'}$ which is between $X(M^{*,f}_{i'})$, $Y(M^{*,f}_{i'})$. 
		\item Define $M_{i+1} := M^*_{i,1} \cup \cdots \cup M^*_{i,k_i} \cup M^f_i \cup \cdots \cup M_i^f$. 
	\end{itemize} 
	We note that $M_{i+1}$ is a matching between $X$ and $Y$ because the sets of vertices $X(M^*_{i,j})$ and $X(M^*_{i,j})$ for $j \in [k_i]$, as well as $X(M^{*,f}_{i'})$ and $Y(M^{*,f}_{i'})$ for $i' \in [i]$ are all disjoint given they are defined with respect to a fixed matching $M^*_i$ 
	over disjoint sets of edges. Moreover, $M_{i+1}$ belongs to $(G^{i+1} \cup B^f_i \cup \ldots \cup B^f_1)[X,Y]$ as $H^i_j$ is part of $G^{i+1}$ for $j \in [k_i]$. 
	It thus only remains to bound the size of $M_{i+1}$. 
	
	For all $i' \in [i]$, $M^f_{i'}$ and $M^{*,f}_{i'}$ are the same so there is nothing to do here. For $j \in [k_i]$, we have, 
	\begin{align*}
		\card{M_{i,j}} = \mu\paren{H^i_j[X(M^*_{i,j}),Y(M^*_{i,j})]} &\geq \mu\paren{B^i_j[X(M^*_{i,j}),Y(M^*_{i,j})])} - \alpha' \cdot n  \tag{as $H^i_j$ is a $\alpha'$-matching-cover of $B^i_j$ and by~\Cref{def:matching-cover}} \\
		&=  \card{M^*_{i,j}} - \alpha \cdot n \tag{as $M^*_{i,j}$ is a perfect matching in $B^i_j$ between $X(M^*_{i,j})$ and $Y(M^*_{i,j})$}. 
	\end{align*}
	Thus, 
	\begin{align*}
		\card{M_{i+1}} = \sum_{j=1}^{k_i} \card{M_{i,j}} + \sum_{i'=1}^{i} \card{M^f_{i'}} &\geq \sum_{j=1}^{k_i} (\card{M^*_{i,j}} - \alpha' \cdot n) + \sum_{i'=1}^{i} \card{M^{*,f}_{i'}} = \card{M^*_{i}} - k_i \cdot \alpha' \cdot n, 
	\end{align*}
	concluding the proof. 
\end{proof}

We can now conclude the bound on the approximation ratio of the algorithm. 

\begin{lemma}\label{lem:rs-final-approx}
	\Cref{alg:rs-matching} outputs an $\alpha$-matching cover of any input multi-graph $G$.  
\end{lemma}
\begin{proof}
	Recall that for every $i \in [t]$, $B^f_i$ denotes the final content of the buffer $B_i$. Moreover by~\Cref{clm:k-i}, buffer $B_t$ never gets full and thus $B^f_t = G^t$. Finally, 
	the algorithm returns $H := (B^f_1,\ldots,B^f_t)$. Fix any disjoint sets of vertices $X,Y \subseteq V(G)$. We have, 
	\begin{align*}
	    \mu(H[X,Y]) &= \mu\paren{(B^f_t \cup B^f_{t-1} \cup \ldots \cup B^f_1)[X,Y]} \tag{by the definition of $H$} \\
	    &= \mu\paren{(G^t \cup B^f_{t-1} \cup \ldots \cup B^f_1)[X,Y]} \tag{as $B^f_t = G^t$} \\
		&\geq \mu\paren{(G^{t-1} \cup B^f_{t-2} \cup \ldots \cup B^f_1)[X,Y]} - k_{t-1} \cdot \alpha' \cdot n \tag{by~\Cref{lem:rs-approx} for $i=t-1$} \\
		&\geq \mu(G[X,Y]) - \sum_{i=1}^{t-1} k_i \cdot \alpha' \cdot n \tag{by repeatedly applying~\Cref{lem:rs-approx} for all $i  < t-1$ and since $G^1=G$} \\
		&\geq \mu(G) - \sum_{i=1}^{t-1} (k/2^{i-1}) \cdot \alpha' \cdot n \tag{by~\Cref{clm:k-i}, $k_i \leq k/2^{i-1}$} \\
		&\geq \mu(G) - 2k \cdot \alpha' \cdot n \tag{by the sum of the geometric series} \\
		&= \mu(G) - \alpha \cdot n \tag{by the choice of $\alpha' = \alpha/2k$}.
	\end{align*}
	This implies that for every disjoint subsets of vertices $X,Y \subseteq V(G)$, we have $\mu(H[X,Y]) \geq \mu(G[X,Y]) - \alpha \cdot n$, thus 
	making $H$ an $\alpha$-matching cover of $G$ by~\Cref{def:matching-cover}. 
\end{proof}

\begin{proof}[Proof of~\Cref{prop:streaming-mc}]
    The bound on the space complexity of the algorithm follows from~\Cref{lem:rs-space} by plugging the value of $\alpha'=\alpha/2k$ and $t=\log{k}+1$. The correctness follows from~\Cref{lem:rs-final-approx}. Finally,~\Cref{alg:rs-matching} is deterministic modulo any potential randomness used by $\MatchingCover$.  
\end{proof}

\subsection{A Streaming Matching Algorithm via Regularity Lemma}\label{sec:stream-matching-regularity}

We now use~\Cref{prop:streaming-mc} together with our~\Cref{thm:matching-cover-polytime} to formalize \Cref{res:stream1} as follows.

\begin{theorem}[Formalization of~\Cref{res:stream1}]\label{thm:stream1}
There is a randomized single-pass streaming algorithm
that with high probability computes a $(1-o(1))$-approximate matching
of a graph presented in a stream with adversarial order of edge arrivals
in $n^2/(\log^* n)^{\Omega(1)}$ space and polynomial time. 
\end{theorem}
\begin{proof}
    To apply~\Cref{prop:streaming-mc}, we need a subroutine $\MatchingCover$ for computing an $(\alpha/2k)$-matching cover (for parameters $\alpha$ and $k$ to be determined soon) on any $n$-vertex
    graph with $n^2/k$ edges.~\Cref{thm:matching-cover-polytime} provides such an algorithm with parameters 
    \[
        (\alpha/2k) = \frac{1}{(\logstar{n})^{\delta_1}} \quad \text{and} \quad \MC{n}{\alpha/2k} = \frac{n^2}{(\logstar{n})^{\delta_2}}, 
    \]
    for some absolute constants $\delta_1,\delta_2 \in (0,1)$. 
    Let $\alpha = 1/(\logstar{n})^{3\delta_1/4}$ and $k=\frac{1}{2} \cdot (\logstar{n})^{\delta_1/4}$, which satisfies the conditions above. 
    Moreover, by~\Cref{prop:regularity-space-efficient}, we can implement~\Cref{alg:matching-cover} of~\Cref{thm:matching-cover-polytime} in polynomial time and space $O((n^2/k) \cdot \log{k}) = n^{2}/(\logstar{n})^{\Omega(1)}$, given only adjacency matrix access to its input graph. This way, by~\Cref{prop:streaming-mc}, we obtain a single-pass streaming algorithm that with high probability computes an $\alpha$-matching cover in space $n^2/(\logstar{n})^{\Omega(1)}$. 
    
    The main algorithm in the theorem is as follows. We store the first $2n^2/k$ edges in the stream using succinct dynamic dictionary of~\Cref{prop:succinct-dict} in $n^2/(\logstar{n})^{\Omega(1)}$ space. In parallel, we also run the algorithm mentioned above to obtain an $\alpha$-matching cover of $G$. The space complexity and polynomial runtime of the algorithm is thus already established. 
    
    We now prove the correctness. If $\mu(G) \leq n/k$, then by~\Cref{fact:m-mu(G)}, we have stored all edges of the graph and thus at the end can simply return a maximum matching of the stored edges; to do so, we simply run Hopcroft-Karp algorithm~\cite{HopcroftK73} by providing it with the adjacency matrix of the stored edges using $\member$~query on the succinct dynamic dictionary (which only requires $O(n\log{n})$ additional space beside the input). Thus, in this case, we obtain an exact maximum matching of the input graph. 
    
    If $\mu(G) > n/k$, then we can pick $X$ and $Y$ in the definition of matching cover output by the algorithm of~\Cref{prop:streaming-mc} to be the endpoints of the maximum matching of $G$, and have, 
    \[
        \mu(H) \geq \mu(G) - \alpha \cdot n \geq (1-\alpha \cdot k) \cdot \mu(G) = (1-1/(\logstar{n})^{\delta_1/2}) \mu(G), 
    \]
    which is $(1-o(1)) \cdot \mu(G)$ as desired. This concludes the proof.
\end{proof}

\subsection{A Streaming Matching Algorithm via RS Graph Upper Bounds}\label{sec:stream-matching-rs}

We formalize~\Cref{res:stream2} as follows in this subsection ($\RS{n}{\beta}$ below was defined in~\Cref{def:rs}). 

\begin{theorem}[Formalization of~\Cref{res:stream2}]\label{thm:stream2}
There exists an absolute constant $\eta > 0$ such that the following is true. There is a randomized single-pass streaming algorithm
that for any $1 \leq k \leq n$ and $\eps \in (0,1/100)$, with high probability,  computes a $(1-\eps)$-approximate matching
of a graph presented in a stream with adversarial order of edge arrivals
in exponential time and space 
\[
		O\Paren{\frac{n^2}{k} \cdot \log^2{k} + \RS{n}{\eta \cdot \eps^2/k} \cdot \log{\paren{\frac{n^2}{\RS{n}{\eta \cdot \eps^2/k}}}} \cdot \log^2{k} \cdot \log{(k/\eps)}}.
\]
Moreover, the algorithm can return an \underline{additive} $\eps \cdot n$ approximation deterministically in exponential time and space
\[
		O\Paren{\frac{n^2}{k} \cdot \log{k} + \RS{n}{\eps/16k} \cdot \log{\paren{\frac{n^2}{\RS{n}{ \eps/16k}}}} \cdot \log{k} \cdot \log{(k/\eps)}}.
\]
\end{theorem}

Roughly speaking, by ignoring lower order terms and in asymptotic notation,~\Cref{thm:stream2} 
gives a streaming algorithm for $(1-o(1))$-approximation of matchings in a single pass with adversarial order of edge arrivals using essentially $(n^2/k + \RS{n}{o(1/k)})$ space for any integer $k \geq 1$.

Before proving~\Cref{thm:stream2}, let us present a corollary of this theorem with concrete bounds on the space by using Fox's triangle removal lemma (\Cref{prop:triangle-removal}) to bound the RS-graph density terms in~\Cref{thm:stream2} (this appears to be the only known method for bounding density of RS graphs with $o(n)$-size induced matchings; moreover, we are not aware of any reference that bounds the density of the type of RS graphs we need, thus we present a proof of that here also for completeness). 

\begin{corollary}\label{cor:stream2}
    There is a deterministic single-pass streaming algorithm
that computes a $(1-o(1))$-approximate matching
of a graph presented in a stream with adversarial order of edge arrivals 
in $n^2/2^{\Omega(\log^* n)}$ space and exponential time. 
\end{corollary}

We prove~\Cref{cor:stream2} in~\Cref{sec:cor-stream2} after proving~\Cref{thm:stream2}. To continue, we need to recall some additional tools from prior work.  specific specific to our algorithms in this subsection. 

\subsubsection{Additional Tools from Prior Work}

\paragraph{Matching covers via RS graphs.} \citet*{GoelKK12} showed that matching covers and RS graphs are intimately connected: on bipartite graphs, the density of best construction for either can be bounded by the density of other one for closely related parameters. We need this result for general graphs as well which follows from the result of~\cite{GoelKK12} using a simple argument\footnote{We can in fact prove this result with better bounds nearly matching those of~\cite{GoelKK12} using a white-box application of the techniques in~\cite{GoelKK12}; however, since the actual constants do not matter for our application in this paper, we opted for the simpler and more direct proof that uses the result of~\cite{GoelKK12} in a black-box way.}.

\begin{proposition}[an extension of~{\cite[Theorem 9.2]{GoelKK12}} to general graphs]\label{prop:matching-cover-rs}
	For any $\alpha \in (0,1)$ and $n \geq 1$, there exists an $\alpha$-matching cover of any $n$-vertex graph with number of edges bounded by 
	\[
		\MC{n}{\alpha} \leq \RS{n}{\alpha/8} \cdot O(\log{(1/\alpha)}).
	\]
\end{proposition}
\begin{proof}

The result of~\cite{GoelKK12} is formally as follows (to match the definitions in our paper, our formulation is slightly different from the statements in~\cite{GoelKK12} but they are equivalent):  
\begin{itemize}[leftmargin=10pt]
    \item[]\cite[Theorem 9.2]{GoelKK12}: For any bipartite graph $G'=(L',R',E')$ with $n$ vertices on \emph{each} side and $\alpha' \in (0,1)$, there exists a subgraph $H'$ with $\RS{2n}{3\alpha'/4} \cdot O(\log(1/\alpha'))$ edges such that for any disjoint subsets of vertices $X \subseteq L'$ and $Y \subseteq R'$, 
    \[
        \mu(H'[X,Y]) \geq \mu(G'[X,Y]) - \alpha' \cdot (2n).
    \]
\end{itemize}
    
    We now use this to prove the bound for general graphs as well. Let $G=(V,E)$ be any (not necessarily bipartite) graph. Consider the bipartite double cover of $G$ obtained by copying vertices of $G$ twice into sets $V_1$ and $V_2$ and connecting any vertex $u_1 \in V_1$ to $v_2 \in V_2$ iff $(u,v)$ is an edge in $G$. Let $G'$ denote this graph and so $G'$ is a bipartite graph with $n$ vertices on each side. 
    
    Compute an $\alpha'$-matching cover $H'$ of this bipartite graph
    using Theorem 9.2 of~\cite{GoelKK12} for parameter $\alpha' = \alpha/2$ (for $\alpha$ given to us in the proposition statement). Thus, $H'$ contains 
    $\RS{2n}{3\alpha'/4} \cdot O(\log(1/\alpha'))$ edges. Create a subgraph $H$ (not necessarily bipartite) on the same vertices as $G$ by adding the edges $(u,v)$ to $H$ iff either $(u_1,v_2)$ or $(v_1,u_2)$ was an edge in $H'$. This way, the number of edges in $H$ will be at most 
    \[
        \RS{2n}{3\alpha'/4} \cdot O(\log{(1/\alpha')}) \leq \RS{n}{\alpha'/4} \cdot O(\log{(1/\alpha')}),
    \]
    where the inequality is by~\Cref{clm:rs-move} that relates density of RS graphs with similar parameters. 
    
    We now argue that $H$ is an $\alpha$-matching cover of $G$. Fix any disjoint subsets of vertices $X,Y$ in $G$. Consider $X_1 \subseteq V_1$ and $Y_2 \subseteq V_2$ corresponding to these two subsets over vertices of $G'$ (and $H'$):
    \begin{align*}
    \mu(H'[X_1,Y_2]) &\geq \mu(G'[X_1,Y_2]) - \alpha' \cdot (2n) \tag{by~\Cref{def:matching-cover} as $H'$ is an $\alpha$-matching cover of $G'$} \\
    &\geq \mu(G[X,Y]) - \alpha' \cdot (2n),
    \end{align*}
    as by the construction of $G'$ any edge $(u,v) \in G[X,Y]$ also has a copy $(u_1,v_2) \in G'[X_1,Y_2]$ and thus $\mu(G'[X_1,Y_2]) \geq \mu(G[X,Y]).$
    Moreover, since $X$ and $Y$ are disjoint, the endpoints of the maximum matching in $H'[X_1,Y_2]$ are disjoint from each other; thus, they are mapped to \emph{unique} edges in $H$ also between $X$ and $Y$, implying that
    \[
        \mu(H[X,Y]) = \mu(H'[X_1,Y_2]) \geq \mu(G[X,Y]) - 2\alpha' \cdot n.
    \]
    Noting that $\alpha' = \alpha/2$ in the above equations, concludes the proof. 
\end{proof}

\paragraph{Vertex-sparsification for matchings.} We also use the reductions of~\citet*{AssadiKLY16} and~\citet*{ChitnisCEHMMV16} for reducing the number of vertices while preserving maximum matching size approximately.  The original versions of the reductions in these work only achieved constant probability of success and boost this to a high probability bound by applying it $\Theta(\log{n})$ times in parallel.  In our setting, we cannot afford this direct success amplification. 
Thus, we instead use the following variant proven by~\citet*{AssadiKL16ec} that achieves a high success probability directly.

\begin{proposition}[{\cite[Lemma 3.8]{AssadiKL16ec}}; see also~\cite{AssadiKLY16,ChitnisCEHMMV16}]\label{prop:vertex-sparsification}
	For any graph $G=(V,E)$, integer $\topt \geq 1$, and parameter $\theta \in (0,1)$, uniformly at random pick a function $h: V \rightarrow [8 \cdot \topt/\theta]$. Consider this \underline{multi-graph} $H=(V_H,E_H)$ obtained from $G$ and $h$: 
	\begin{itemize}
		\item $V_H$ is the range of the function of $h$, thus $\card{V_H} = 8 \cdot \topt/\theta$.  
		\item For any edge $(u,v) \in G$, there is an edge $(h(u),h(v)) \in E_H$. 
	\end{itemize}
	If $\mu(G) \leq \opt$, then, 
	\[
		\Pr_{h}\Paren{\mu(H) < (1-\theta) \cdot \mu(G)} \leq \exp\paren{-\frac{\mu(G)}{4}}. 
	\]
\end{proposition}

\subsubsection{Proof of~\Cref{thm:stream2}}

We now use these prior tools combined with our~\Cref{prop:streaming-mc} to prove~\Cref{thm:stream2}. Recall that~\Cref{prop:streaming-mc} returns an $\alpha$-matching cover which can only guarantee an additive approximation not a multiplicative one. Thus, we first use the vertex-sparsification of~\Cref{prop:vertex-sparsification} to reduce the number of vertices in $G$ to $O(\mu(G))$---by guessing $\mu(G)$ in geometric values---so that an additive approximation also becomes a multiplicative one. We then use~\Cref{prop:matching-cover-rs} to compute the matching covers in~\Cref{alg:rs-matching} of~\Cref{prop:streaming-mc}.  

\begin{Algorithm}\label{alg:rs-final}
	\begin{center}
		The randomized algorithm in~\Cref{thm:stream2}. 
	\end{center}
	\medskip
	
	\textbf{Input:} A graph $G=(V,E)$ in the stream with $n$ edges and at most ${{n}\choose{2}}$ edges. We are also given integer $k \geq 1$ and approximation parameter $\eps \in (0,1)$ 
	as in~\Cref{thm:stream2}. 
	
	\smallskip
	
	\textbf{Output:} A $(1-\eps)$-approximate maximum matching of $G$. 
	
	\smallskip
	
	\begin{enumerate}[label=$(\roman*)$]
		\item For $i=1$ to $t:=\log{k}$ iterations in parallel: 
		\begin{enumerate}
			\item Let $\topt_i := n/2^{i+1}$ and pick a hash function $h_i : V \rightarrow [32 \cdot \topt_i/\eps]$. 
			\item Consider the multi-graph $G_i$ obtained from $G$ and $h_i$ 
			using~\Cref{prop:vertex-sparsification}; each edge of $G$ arriving in the stream can be mapped to an edge of $G_i$ using $h_i$. 
			\item Run~\Cref{alg:rs-matching} on $G_i$ with parameters $k$ and $\alpha=\eps^2/64$ and $m={{n}\choose{2}}$ to obtain an $\alpha$-matching cover $H_i$. We use the matching cover construction of~\Cref{prop:matching-cover-rs} as the subroutine $\MatchingCover$ (as specified in~\Cref{clm:rs-final-space} below). 
		\end{enumerate}
		\item Store the first $n^2/k$ edges of the stream using succinct dynamic dictionary of~\Cref{prop:succinct-dict} as the subgraph $H_0$. 
		
		\item Return a maximum matching in $H_0 \cup H_1 \cup \ldots \cup H_t$ (specified in~\Cref{clm:rs-final-space} below). 
	\end{enumerate}
\end{Algorithm}
\noindent
We bound the space and approximation of \Cref{alg:rs-final} in the following two claims, respectively. 

\begin{claim}\label{clm:rs-final-space}
	\Cref{alg:rs-final} (deterministically) requires space of 
	\[
		O(\frac{n^2}{k} \cdot \log^2{k} + \RS{n}{\eps^2/1024k} \cdot \log\paren{\frac{n^2}{\RS{n}{ \eps^2/1024k}}} \cdot \log^2{k} \cdot \log{(k/\eps)}).
	\]
\end{claim}
\begin{proof}
    Consider each iteration $i \in [t]$. We have a multi-graph with $n_i := \min\set{32\opt_i/\eps, n}$ vertices (since $G$ has $n$ vertices, $G_i$ cannot have more than $n$ vertices with non-zero degrees and we can ignore the remaining vertices without loss of the generality). We are using \Cref{prop:streaming-mc} with subroutine $\MatchingCover$ that finds an $(\alpha/2k)$-matching cover using~\Cref{prop:matching-cover-rs} (we specify how this step is implemented below). This implies that the size of the matching cover is
    \[
        \MC{n_i}{\alpha/2k} = \RS{n_i}{\alpha/16k} \cdot O(\log{(k/\alpha)}) = \RS{n}{\eps^2/1024k} \cdot O(\log{(k/\eps)}),
    \]
    as $n_i \leq n$. As such, since $m \leq n^2$, by~\Cref{prop:streaming-mc}, each iteration requires space of:  
    \[
    	O\Paren{\frac{n^2}{k} \cdot \log{k} + \RS{n}{\eps^2/{1024k}} \cdot \log{(k/\eps)} \cdot \log{\paren{\frac{n^2}{\RS{n}{{\eps^2}/{1024k}}}}} \cdot \log{k}};
    \]
    Given we have $O(\log{k})$ iterations, this concludes the bound on the space of the algorithm (storing $O(n^2/k)$ edges in step $(ii)$ using~\Cref{prop:succinct-dict} requires another $O((n^2/k) \cdot \log{k})$ bits). 
    
 Finally, we make sure $\MatchingCover$ as well as step $(iii)$ of the algorithm can be implemented in this space. For $\MatchingCover$, we need an $O((n_i^2/k) \cdot \log{k})$ space algorithm for finding an $(\alpha/2k)$-matching cover of a graph with $n_i^2/k$ edges with $\RS{n_i}{\alpha/16k} \cdot O(\log{(k/\alpha)})$ edges, whose existence is promised by~\Cref{prop:matching-cover-rs}. To obtain this, we simply enumerate over all subsets of edges in the input graph to $\MatchingCover$, and then enumerate over all subsets of vertices to check whether this subset is a matching cover; for each subset also, we run Hopcroft-Karp algorithm~\cite{HopcroftK73} to compute the size of the matching in the input graph and subset of edges as a potential cover, to ensure this subset can be a matching cover.
 
 Furthermore, all of this is done by storing intermediate edges in a succinct dynamic data structure of~\Cref{prop:succinct-dict} (with its deterministic guarantee as we ignore the runtime since our algorithm is exponential time anyway). This requires using $O(n^2/k \cdot \log{k})$ space in total. Finally, step $(iii)$ can also be implemented again by running Hopcroft-Karp algorithm~\cite{HopcroftK73} over adjacency matrix of the stored edges provided by \member~query to~\Cref{prop:succinct-dict} for these edges. 
\end{proof}

\begin{claim}\label{clm:rs-final-approx}
	\Cref{alg:rs-final} outputs a $(1-\eps)$-approximate matching with high probability. 
\end{claim}
\begin{proof}
	Suppose first that $\mu(G) \leq n/2k$. By~\Cref{fact:m-mu(G)}, $G$ in this case 
	has at most $2n \cdot \mu(G) \leq n^2/k$ edges. Thus, in step $(ii)$ of the algorithm, we are simply storing all edges and thus the algorithm returns an exact answer. 
	
	\noindent
	Now suppose $\mu(G) > n/2k$. This means that there is an index $i \in [t]$ such that 
	\[
	\frac{n}{2^{i+1}} \leq \mu(G) < \frac{n}{2^i}.
	\] 
	For this choice of $i$, we have $\topt_i \leq \mu(G) < 2 \cdot \topt_i$ (and $\mu(G) > n^2/2k \geq n/2$). By~\Cref{prop:vertex-sparsification} for $\theta=\eps/2$ and $\topt = 2 \cdot \topt_i > \mu(G)$, 
	and $h_i : V \rightarrow [32\topt_i/\eps]$, we have, 
	\[
		\Pr_{h_i}\paren{\mu(G_i) < (1-\eps/2) \cdot \mu(G)} \leq \exp\paren{-\frac{\mu(G)}{4}} \ll 1/\poly{(n)},
	\]
	 where we used that fact $32\topt_i/\eps = 8 \cdot \topt/\theta$. 
	We condition on the complement of this event which happens with high probability. Based on this, we further have that 
	\[
		n_i := \card{V(G_i)} = \frac{32}{\eps} \cdot {\topt_i} \leq \frac{32}{\eps} \cdot \mu(G).
	\]
    Since $H_i$ is an $\alpha$-matching cover of $G_i$, by letting $X$ and $Y$ in~\Cref{def:matching-cover} to be the endpoints of the maximum matching of $G_i$, we have, \[  
        \mu(H_i) \geq \mu(G_i) - \alpha \cdot n_i \geq (1-\eps/2) \cdot \mu(G) - (\eps^2/64) \cdot \frac{32}{\eps} \cdot \mu(G) = (1-\eps) \cdot \mu(G). 
    \]
    Thus, returning the maximum matching of $H_i$ as part of $H_0 \cup \ldots \cup H_t$ achieves a $(1-\eps)$-approximation, concluding the proof. 
\end{proof}

\Cref{thm:stream2} for randomized case now follows from~\Cref{clm:rs-final-space,clm:rs-final-approx}. For the deterministic part with additive approximation guarantee, we simply forgo guessing $\mu(G)$ and using vertex-sparsification of~\Cref{prop:vertex-sparsification} at all; instead, we just run~\Cref{alg:rs-matching} over the entire input and use~\Cref{prop:matching-cover-rs}, the same way as above exactly, as the subroutine $\MatchingCover$ for computing an $\alpha$-matching cover. Since we now only need an additive $\eps \cdot n$ guarantee, we can take $\alpha = \eps$ directly which implies the improved bounds on the space as well.

\subsubsection{Proof of~\Cref{cor:stream2}}\label{sec:cor-stream2}

We are now going to prove~\Cref{cor:stream2} by explicitly upper bounding the $\textsf{RS}$ term in~\Cref{thm:stream2}. To do so, we need the following lemma on density of RS graph. The proof of this lemma uses standard ideas but we are not aware of any reference that explicitly states this bound, hence we prove it here for completeness. 

\begin{lemma}\label{lem:RS-density}
    For any integer $n \geq 1$ and constant $c \in (0,1)$
    \[
        \RS{n}{\frac{c}{2^{\,\logstar{n}}} }\leq \frac{n^2}{2^{\,\logstar{n}}}. 
    \]
\end{lemma}
\begin{proof}
    As the bipartite double cover of any RS graph is also a bipartite RS graph, we can assume without loss of generality that $\RS{n}{\beta}$ for any $\beta$ corresponds to the density of some bipartite RS graph. Thus, in the following, we only work with bipartite RS graphs.  
    
    Let $G$ be the densest possible $(r,t)$-RS bipartite graph on $n$ vertices with $t$ induced matchings $M_1,\ldots,M_t$ each of size $r = c \cdot n/2^{\logstar{n}}$ so that we have $r \cdot t = \RS{n}{c/2^{\logstar{n}}}$. 
    Suppose towards a contradiction that $t \geq n$ as otherwise $r \cdot t = c \cdot n^2/2^{\logstar{n}}$ already.

    Define $r' = 8 \cdot 2^{(\logstar{n})/b}$ for the constant $b > 1$ in the triangle removal lemma (\Cref{prop:triangle-removal}). Note that $r' < r$ (for sufficiently large $n$ as $c$ is a constant). In the following, we pick $r'$ arbitrary edges from each of $M_1,\ldots,M_t$ and discard the remaining edges to obtain an $(r',t)$-RS graph. Based on this, we define the following graph:
    \begin{itemize}
        \item For any induced matching $M_i$ of size $r'$, add a new vertex $z_i$ and connect it to both endpoints of any edge in $M_i$ in $G$;
        \item Call the resulting graph on these $2n$ vertices $H$.
    \end{itemize}
    
    We claim that $H$ has precisely $r' \cdot t$ triangles: this is because $G$ was bipartite and each $M_i$ is an induced  matching, so each newly added vertex $z_i$ can create precisely $r'$ triangles. 
    At the same time, to make $H$ triangle free, we need to remove one edge from each triangle $(z_i,u,v)$ for each $(u,v) \in M_i$ as these triangles are edge disjoint. Thus, we need to remove $r' \cdot t$ edges from $H$ to make it triangle free.  
    
    Define
    \[
    \gamma := \frac{r' \cdot t}{2 \cdot (2n)^2} = \frac{r' \cdot n}{8n^2} = \frac{1}{2^{(\logstar{n})/b}};
    \]
    thus, we know that strictly more than $\gamma \cdot (2n)^2 < r' \cdot t$ edges of $H$ need to be removed before it becomes triangle free. Further define $\delta \in (0,1)$ such that 
    \[
        \delta^{-1} = 2 \upuparrows b \cdot \log{(1/\gamma)} = 2 \upuparrows \logstar{n} = n.
    \]
    Given that any $(2n)$-vertex graph with $\delta \cdot (2n)^3$ triangles can be made triangle free by removing $\gamma \cdot (2n)^2$ edges, while $H$ cannot (by the choice of $\gamma$), we have that the number of triangles in $H$ needs to be more than $\delta \cdot (2n)^3$, which implies that 
    \[
        r' \cdot t > \delta \cdot (2n)^3 = 8 \cdot n^2.
    \]
    But this is a contradiction since $r' \cdot t \leq n^2$ as  $r' \cdot t$ is the density of a $(r',t)$-RS graph on $n$ vertices and no $n$-vertex 
    (simple) graph can have more than $n^2$ edges. 
    This implies that our original assumption that $t \geq n$ was false, concluding the proof. 
\end{proof}

We can now conclude the proof of~\Cref{cor:stream2}. 

\begin{proof}[Proof of~\Cref{cor:stream2}]
    The algorithm stores the first $n^2/2^{(\logstar{n})/4}$ edges of the stream. This is done using succinct dynamic dictionary of~\Cref{prop:succinct-dict} in $n^{2}/2^{\Omega(\logstar{n})}$ space. Thus, if $\mu(G) < n/2^{(\logstar{n})/4+1}$, by~\Cref{fact:m-mu(G)}, we have stored all edges of the graph and can solve the problem exactly. 
    
    Otherwise, we set $k = 2^{(\logstar{n})/4}/4$ and $\eps = 1/2^{(3\logstar{n})/4}$ (so $\eps/k = 1/2^{\logstar{n}}$) in (moreover part of)~\Cref{thm:stream2} and obtain a deterministic algorithm 
    with $\eps \cdot n$ additive approximation guarantee
 with space 
    \[
        \frac{n^2}{2^{\Omega(\logstar{n})}} + \RS{n}{\,{(1/16)}/{2^{(\logstar{n})}}} \cdot \log{\paren{\frac{n^2}{\RS{n}{\,{(1/16)}/{2^{(\logstar{n})}}}}}} \cdot  (\logstar{(n)})^{O(1)}. 
    \]
 The above term can be bounded by $n^2/2^{\Omega(\logstar{n})}$, since ~\Cref{lem:RS-density} implies that
 \[
 \RS{n}{(1/16)/2^{\logstar{n}}} \leq n^2/2^{\Omega(\logstar{n})}. 
 \]

    Finally, the returned matching has size 
    \[
        \mu(G) - \eps \cdot n \geq \mu(G) - \eps \cdot \mu(G) \cdot 2^{(\logstar{n})/4+1} = (1-2^{-(\logstar{n})/2+1}) \cdot \mu(G) = (1-o(1)) \cdot \mu(G)
    \]
    where the inequality is by the lower bound on $\mu(G)$ and its next equality is by the choice of $\eps$.  This concludes the proof. 
\end{proof}

\subsection*{Acknowledgments}
Sepehr Assadi would like to thank Huacheng Yu for helpful discussions regarding succinct dynamic dictionaries. 

\clearpage

\bibliographystyle{plainnat}
\bibliography{references}
	
\clearpage
\appendix

\section{Missing Proofs of Preliminary Results in~\Cref{sec:prelim}}\label{app:prelim}

\subsection{Proof of~\Cref{prop:regularity-space-efficient}}\label{app:prop-regularity-space-efficient}

\begin{proposition*}
  Given query access to the adjacency matrix,
  the algorithm in~\Cref{thm:regularityalgo} can be implemented
  in $O(n \cdot Q(t,1/\gamma)\log n)$ space and $\poly(n, Q(t,1/\gamma))$ time.
\end{proposition*}
\begin{proof}[Proof Sketch]
We now briefly describe how the algorithm in~\Cref{thm:regularityalgo} can be implemented in a space-efficient manner, given query access to the adjacency matrix of the underlying graph $G$.
Roughly speaking, the algorithm works as follows.
Initially, it starts with an arbitrary equitable $t$-partition.
As long as more than $\gamma$ fraction of the $C_i,C_j$ pairs are not $\gamma$-regular,
for each such irregular pair $C_i,C_j$, a {\em witnessing pair} $X\subseteq C_i, Y\subseteq C_j$ is identified such that
$X,Y$ violate the regularity property. Then the algorithm does a refinement of the partition such that simultaneously for all irregular pairs $C_i,C_j$ with witnessing pair $X,Y$, vertices in $X$ vs. $C_i\setminus X$ and those in $Y$ vs. $C_j\setminus Y$ are separated.
A potential function argument then shows that there cannot be more than $\poly(t,1/\gamma)$ refinements before we obtain a $\gamma$-regular partition.

Note here that the refinement of the partition is easy to implement in $O(nk\log n)$ space and $\poly(n)$ time,
with $k$ being the number of classes, as it only requires storing a description of the vertex partition.
It then remains to analyze the process of finding a witness $X,Y$ for each irregular pair $C_i,C_j$.
This is done in~\cite{Alon02} by an approximation algorithm, where the main step requires
computing the number of common neighbors for each vertex pair $u,v\in V$,
by squaring the adjacency matrix via fast matrix multiplication.
However, this can be easily done in $O(n\log n)$ space and $\poly(n)$ time given query access to the adjacency matrix.
Thus the entire algorithm can be implemented in $O(n \cdot Q(t,1/\gamma)\log n)$ space and $\poly(n, Q(t,1/\gamma))$ time. 
\end{proof}

\subsection{Proof of~\Cref{clm:rs-move}}\label{app:prop-rs-move}

\begin{claim*}
    For any integer $n \geq 1$ and real number $0 < \beta < 1$,  $\RS{2n}{3\beta} \leq O(1) \cdot \RS{n}{\beta}$. 
\end{claim*}
\begin{proof}
    If $\beta < 1/\sqrt{n}$, then both $\RS{2n}{2\beta} = (1-o(1)) \cdot {{2n}\choose{2}}$ and $\RS{n}{\beta} = (1-o(1)) \cdot {{n}\choose{2}}$ (see, e.g.,~\cite{AlonMS12}), which satisfy the claim bounds. We now prove the case when $\beta \geq 1/\sqrt{n}$.

    Fix any $(r,t)$-RS graph $G$ on $2n$ vertices with $r = (3\beta) \cdot (2n)$ and $r \cdot t = \RS{2n}{3 \beta}$. We use $G$ to construct an $(r',t)$-RS graph $H$ on $n$ vertices with $r' = (\beta) \cdot n$ and $r' \cdot t \geq \Omega(1) \cdot (r \cdot t)$; this implies that $\RS{n}{\beta} \geq \Omega(1) \cdot \RS{2n}{3\beta}$, as desired. 
    
    To construct $H$, sample exactly half the vertices of $G$ uniformly at random and add all edges in $G$ between the sampled vertices. This way $H$ has $n$ vertices. Moreover, for any induced matching $M$ in $G$, we have sampled $1/4$ of its edges in expectation, and thus at least $1/5$ with high probability (using Chernoff bound for sampling without replacement and since $\beta \cdot n \geq \sqrt{n}$). Thus, each induced matching now has size at least $3\beta \cdot 2n/5 = 6/5 \cdot \beta \cdot n > \beta \cdot n$. Moreover, the number of edges in $H$ is again with high probability at least $1/5$ of the edges in $G$. Thus, we can remove another constant fraction of edges in $H$ so that all induced matchings have size exactly $\beta \cdot n$, and obtain an $(r',t)$-RS graph with the desired parameters, concluding the proof. 
\end{proof}

\end{document}